\documentclass{llncs}

\usepackage{fullpage}
\usepackage{times}

\usepackage{url}
\usepackage{color}
\usepackage{enumitem}
\usepackage {xspace}
\usepackage{epsfig}
\usepackage[cmex10]{amsmath}
\usepackage {amsfonts}
\usepackage {float}

\newcommand{\Alg}{\mbox{ALG}\xspace}
\newcommand{\OPT}{\mbox{OPT}\xspace}
\newcommand{\OFF}{\mbox{OFF}\xspace}
\newcommand{\cA}{{\mathcal A}\xspace}

\newcommand{\LISs}{(n,\beta)\mbox{-LIS}\xspace}
\newcommand{\LIS}{\mbox{LIS}\xspace}
\newcommand{\KLIS}{K\mbox{-LIS}\xspace}
\newcommand{\LPT}{\mbox{LPT}\xspace}

\newcommand{\cC}{{\mathcal C}}
\newcommand{\cT}{{\mathcal T}}
\newcommand{\LAF}{\mbox{LAF}\xspace} 
\newcommand{\total}{total}
\renewcommand{\ell}{c}

\newcommand{\lmin}{{\ell_{min}}\xspace}
\newcommand{\lmax}{{\ell_{max}}\xspace}
\newcommand{\gnBurst}{\gamma\mbox{n-Burst}\xspace}
\newcommand{\GroupLIS}{\mbox{GroupLIS}}

\newenvironment{proofof}[1]{\bigskip \noindent {\bf Proof of #1:}}
{\qed\par\vskip 0.5mm\par}

\usepackage[final]{changes}
\definechangesauthor{af}{red}
\definechangesauthor{cg}{blue}
\definechangesauthor{dk}{brown}
\definechangesauthor{ez}{magenta}

\newcommand{\remove}[1]{}

\newcommand{\cg}[1]{{\color{blue} #1}}
\newcommand{\af}[1]{{\color{red} #1}}
\newcommand{\dk}[1]{{\color{green} #1}}
\newcommand{\ez}[1]{{\color{magenta} #1}}

\begin{document}

\title{Online Parallel Scheduling of Non-uniform Tasks: \\
Trading Failures for Energy\thanks{This research was supported in part by the Comunidad de Madrid grant S2009TIC-1692, Spanish MICINN/MINECO grant TEC2011-29688-C02-01, and NSF of China grant 61020106002.}
}

\author{Antonio Fern\'andez Anta\inst{1}
	\and Chryssis Georgiou\inst{2}
	\and Dariusz R. Kowalski\inst{3}
	\and Elli Zavou\inst{1,4}\thanks{Partially supported by FPU Grant from MECD}}

\institute{Institute IMDEA Networks
	\and University of Cyprus
	\and University of Liverpool
	\and Universidad Carlos III de Madrid\vspace{-2em}}


\date{}
\maketitle \thispagestyle{empty}

\begin{abstract}

Consider a system in which tasks of different execution times arrive continuously and have to be executed by a set of processors that are prone to crashes and restarts.
In this paper we model and study the impact of parallelism and failures on the competitiveness of such an online system. In a fault-free environment, a simple Longest-in-System scheduling policy, enhanced by a redundancy-avoidance mechanism, guarantees optimality in a long-term execution.
In the presence of failures though, scheduling becomes a much more challenging task. In particular, \deleted[ez]{we show that} no parallel deterministic algorithm can be competitive against an offline optimal solution, even with one single processor and tasks of only two different execution times. We find that when additional energy is provided to the system in the form of processor speedup, the situation changes. 
Specifically, we identify thresholds on the speedup under which such competitiveness cannot be achieved by any deterministic algorithm, and above which competitive algorithms exist. \replaced[ez]{Finally}{Moreover}, we propose algorithms that achieve small bounded competitive ratios when the speedup is over the threshold.\vspace{.5em}

\vspace{-1em}
\end{abstract}



\section{Introduction\vspace{-.3em}}
\label{sec:intro}

\noindent {\bf Motivation.}
\sloppy{In  recent years we have witnessed a dramatic increase on the demand of processing computationally-intensive jobs. Uniprocessors are no longer capable of coping with the high computational demands of such jobs.} As a result, multicore-based parallel machines such as the K-computer~\cite{KComputer} and Internet-based supercomputing platforms such as SETI@home~\cite{SETI} and EGEE Grid~\cite{EGEE} have become prominent computing environments. 
However, computing in such environments raises several challenges. For example, computational jobs (or tasks) are injected dynamically and continuously, each job may have different computational demands (e.g., CPU usage or processing time) and the 
processing elements
are subject to unpredictable failures. Preserving power consumption is another challenge of rising importance. Therefore, there is a corresponding need for developing algorithmic solutions that would efficiently cope with such challenges. 

\replaced[af]{Much research has been dedicated to task scheduling problems, each work addressing
different challenges (e.g., \cite{AKP_STOC92,CMR_TPDS07,DOOPM_MTAGS10,Emeketal_PODC10,GeorgiouK11,GS_book08,166609,KS97,schedulingbook,DBLP:conf/focs/YaoDS95,CDS_DC01}).}
{For this purpose, much research, spanning the areas of Parallel and Distributed Computing, as well as Scheduling, was dedicated over the last two decades to this problem (e.g., \cite{AKP_STOC92,CMR_TPDS07,DOOPM_MTAGS10,Emeketal_PODC10,GeorgiouK11,GS_book08,166609,KS97,schedulingbook,DBLP:conf/focs/YaoDS95,CDS_DC01}). Different works address the issue by tackling 
different challenges.}
For example, many works address the issue of dynamic task injections, but do not consider failures (e.g., \cite{Chan:2009:SSP:1583991.1583994,160366}). 
Other works consider scheduling on one machine (e.g., \cite{Albers:2012:RIN:2095116.2095216,DBLP:journals/tse/SchwanZ92,5062123}); with the drawback that the power of parallelism is not exploited (provided that tasks are independent). Other works consider failures, but assume that tasks are known a priori and their number is bounded (e.g., \cite{AllocateFOCS,AW_SIAM97,CDS_DC01,GS_book08,KS97}),
where other works assume that tasks are uniform, that is, they have the same processing times (e.g., \cite{CDS_DC01,GeorgiouK11}).
Several works consider power-preserving issues, but do not consider, for example, failures (e.g.,~\cite{Bansal:2009:SSA:1496770.1496846,Chan:2009:SSP:1583991.1583994,DBLP:conf/focs/YaoDS95}).\vspace{.2em}  

\noindent{\bf Contributions.}
In this work we consider a computing system in which tasks of {\em different} execution times arrive {\em dynamically and continuously} and must be performed by a set of $n$ processors that are prone to {\em crashes and restarts}.
Due to the dynamicity involved, we view this task-performing problem as an online problem and pursue competitive analysis~\cite{ST_CACM85,AADW_FOCS94}. 
Efficiency is measured as the maximum {\em pending cost} over any point of the execution,
where the pending cost is the sum of the execution times of the tasks that have been injected in the system but are not performed yet. 
We also account for the maximum {\em number of pending tasks} over any point of the execution. 
The first measure is useful for evaluating the remaining processing time required from the system at any given point of the computation, while the second for evaluating the number of tasks still pending to be performed, regardless of the processing time needed.
\deleted[ez]{Our goal is to explore the impact of parallelism, the different task execution times, and the faulty environment on the competitiveness of the online system we consider.}

We show that no parallel algorithm for the problem under study
is competitive against the best off-line solution in the classical sense, however it becomes competitive if static processor {\em speed scaling}~\cite{AGM-ICALP11,Albers:2011:MSS:1989493.1989539,Chan:2009:SSP:1583991.1583994} is applied in the form of a {\em speedup} above a certain threshold.
A speedup $s$ means that a processor can perform a task $s$ times faster than the task's system specified execution time (and therefore has a meaning only when $s \ge 1$). 
Speed scaling
impacts the {\em energy consumption} of the processor.
As a matter of fact, the power consumed (i.e., the energy consumed per unit of time)
to run a processor at a speed $x$ grows superlinearly with $x$, and it is typically assumed to have a form of $P =  x^\alpha$, for $\alpha>1$ \cite{DBLP:conf/focs/YaoDS95,Intelwhitepaper}.
Hence, a speedup $s$ implies an additional factor of $s^{\alpha-1}$ in the power (and hence energy) consumed.
The use of a speedup is a form of resource augmentation \cite{DBLP:journals/algorithmica/PhillipsSTW02}.

\begin{table}[tdp]
\caption{Summary of results.}
\begin{center}
\begin{tabular}{|c|c|c|c|c|}
\hline
Condition & Task costs & Task competitiveness & Cost competitiveness & Algorithm \\ \hline \hline
$s < \lmax / \lmin$ and $s < \frac{\gamma \lmin+\lmax}{\lmax}$ & $\ge 2$ & $\infty$ & $\infty$ & Any \\ \hline
$s \geq \lmax/\lmin$ & Any & 1 & $\lmax/\lmin$ & $\LISs$ \\ \hline
$\frac{\gamma\lmin + \lmax}{\lmax} \leq s < \lmax / \lmin$ & 2 & 1 & 1 & $\gnBurst$ \\ \hline
$s \geq 7/2$ & Finite & $\lmax / \lmin$ & 1 &  $\LAF$ \\ \hline
\end{tabular}
\end{center}
\label{table-results}\vspace{-2em}
\end{table}%

Our investigation aims at developing competitive online algorithms that require the smallest possible speedup. 
As a result, one of the main challenges of our work is to identify the speedup thresholds, under which competitiveness cannot be achieved and over which it is possible.
In some sense, our work can be seen as investigating the trade-offs between {\em knowledge} and {\em energy} 
in the presence of failures: 
How much 
energy (in the form of speedup) does a deterministic online scheduling algorithm need in order to match the efficiency (i.e., to be competitive with) of the optimal offline algorithm that possesses complete knowledge of failures and task injections? (It is understood that there is nothing to investigate if the offline solution makes use of speed-scaling as well).
Our contributions are summarized as follows \added[af]{(see Table \ref{table-results})}:\vspace{-.5em}
\begin{description}[leftmargin=1mm]
\item [\em Formalization of fault-tolerant distributed scheduling:]
In Section~\ref{sec:model}, we formalize an online task performing problem that abstracts important aspects of today's multicore-based parallel systems and Internet-based computing platforms: dynamic and continuous task injection, tasks with different processing times, processing 
elements
subject to failures, and concerns on power-consumption. 
To the best of our knowledge, this is the first work to consider
such a version of dynamic and parallel fault-tolerant task scheduling.

\item[\em Study of offline solutions:]
In Section~\ref{sec:NPhard}, we show that an offline version of a similar task-performing problem is NP-hard, for both pending cost and pending task efficiency, even if there is no parallelism 
(one processor) and the information of all tasks and processor availability is known.

\item[\em Necessary conditions for competitiveness:]
In Section~\ref{s:non-comp}, we show \emph{necessary} conditions (in the form of threshold values) on the value of the speedup $s$ to achieve competitiveness.
\deleted[ez]{We first consider a rather natural threshold:
$s=\frac{\lmax}{\lmin}$, where $\lmin$ and $\lmax$ are lower and upper bounds on the cost (i.e., running time) of all tasks injected in the system.
This implies that during the time the offline algorithm needs to perform a task of cost $\lmin$, a given online algorithm could even perform a task of cost $\lmax$. As it turns out, this threshold is not enough, and thus, we introduce a second one, which requires a new parameter $\gamma$.}
\replaced[ez]{To do this, we need to introduce a parameter $\gamma$, which represents}{intuitively, $\gamma$ is} the smallest number of $\lmin$-tasks that an algorithm can complete (using speedup $s$), in addition to a $\lmax$-task, such that the offline algorithm cannot complete more tasks in the same time. \added[ez]{Note that $\lmin$ and $\lmax$ are lower and upper bounds on the cost (execution time) of the tasks injected in the system.}
\deleted[ez]{In other words, $\gamma$ is the smallest integer satisfying $\frac{\gamma\lmin+\lmax}{s} \le (\gamma+1)\lmin$ 
(and hence, $\gamma = \max\{\lceil \frac{\lmax-s \lmin}{(s-1)\lmin} \rceil, 0\}$).
We then define the second threshold as $s=\frac{\gamma \lmin+\lmax}{\lmax}$.}

We \deleted[ez]{therefore} propose two conditions, (a) $s < \frac{\lmax}{\lmin}$, and (b) $s < \frac{\gamma \lmin+\lmax}{\lmax}$ \replaced[ez]{and}{. We} show that if \emph{both} \deleted[ez]{conditions (a) and (b)} hold, then
{\em no} deterministic sequential or parallel algorithm is competitive when run with speedup $s$.
\footnote{It is worth noting that this holds even if we only have a single processor, and therefore this result could be generalized for stronger models that use centralized or parallel scheduling of multiple processors.}
Observe that, satisfying condition (b) implies $\rho > 0$, which automatically means that condition (a) is also satisfied.

\item[\em Sufficient conditions for competitiveness:]
Then, we design two scheduling algorithms, each matching a different threshold bound from the necessary conditions above, showing \emph{sufficient} conditions on $s$ that lead to competitive solutions.
In fact, it can be shown that in order to have competitiveness, it is sufficient to set $s=\lmax/\lmin$ if $\lmax/\lmin \in [1,\varphi]$, and $s= 1 + \sqrt{1 - \lmin/\lmax}$ if otherwise, where $\varphi=\frac{1+\sqrt{5}}{2}$, which is the golden ratio (see details in Appendix~\ref{app:conditions}).

{\bf Algorithm $\LISs$:}
For the case when condition (a) does not hold (i.e., $s \geq \frac{\lmax}{\lmin}$), we develop algorithm $\LISs$, presented in Section~\ref{s:alg}. We show that, under these circumstances, $\LISs$ 
\added[dk]{is $1$-pending-task-competitive and $\frac{\lmax}{\lmin}$-pending-cost-competitive}
\deleted[dk]{achieves a bound on the number of pending tasks of $\cT_t(\OPT) + \beta n^2 + 3n$}
for parameter $\beta \geq \frac{\lmax}{\lmin}$ and for any given number of processors $n$. 
\deleted[dk]{where $\cT_t(\OPT)$ denotes the number of pending tasks of the best offline algorithm \OPT at time $t$. 
The pending-cost-competitiveness of $\LISs$, for the same speedup $s$ and parameter $\beta$, is 
$\frac{\lmax}{\lmin}\cdot\left(\cC_t(\OPT) + \beta n^2+3n\right)$, where $\cC_t(\OPT)$ is the pending cost of $\OPT$ at time $t$.}
\deleted[ez]{In order to be more cost competitive, $\LISs$ needs to be used with a higher speedup, $s \geq \frac{\lmax}{\lmin}\cdot \lceil\frac{\lmax}{\lmin}\rceil$, for which we show a bound on pending cost of $\cC_t(\OPT) + \lmax\beta n^2 + (2\lmax+\lmin)n$ for $\beta \geq \frac{\lmax}{\lmin}$ and any $n$.}
These results hold for any collection of tasks with costs in the range $[\lmin,\lmax]$.

{\bf Algorithm $\gnBurst$:}
It is not difficult to observe that algorithm $\LISs$ cannot be competitive when condition (a) holds but condition (b) does not (i.e., $\frac{\gamma\lmin + \lmax}{\lmax} \leq s < \frac{\lmax}{\lmin}$). For this case we develop algorithm $\gnBurst$, presented in Section~\ref{sec:gnBurst}. We show that when tasks of two different costs, $\lmin$ and $\lmax$, are injected, the algorithm is 
\replaced[dk]{both $1$-pending-task and $1$-pending-cost competitive.}{competitive.}
\deleted[dk]{
Namely, it holds that $\cT_t(\gnBurst) \leq \cT_t(\OPT) + 2n^2 + (3+\left\lceil \frac{\lmax}{s  \cdot \lmin} \right\rceil)n$, and $\cC_t(\gnBurst) \leq \cC_t(\OPT) +  \lmax(n^2 + 2n) + \lmin(n^2 + (1+\left\lceil \frac{\lmax}{s  \cdot \lmin} \right\rceil)n)$.
}

These results fully close the gap with respect to the conditions for competitiveness on the speedup in the case of two different task costs, establishing $s=\min\{\frac{\lmax}{\lmin}, \frac{\gamma \lmin+\lmax}{\lmax}\}$ as the threshold for competitiveness.
(A detailed analysis of \replaced[dk]{its exact value}{the exact value of this threshold} can be found in Appendix~\ref{app:conditions}.)

\item[\em Algorithm $\LAF$, low energy guaranteed:]
In Section~\ref{sec:LAF}, we develop algorithm $\LAF$ that is again competitive for the case when condition (b) does not hold, but in contrast with $\gnBurst$, it is more ``geared'' towards pending cost efficiency and can handle tasks of multiple different costs.
We show that this algorithm is competitive for speedup \replaced[af]{$s \geq \frac{7}{2}$}{$s > \frac{7}{2}$}. Hence,
unlike the above mentioned algorithms, its competitiveness is with respect to a speedup that is independent of the values 
$\lmax$ and $\lmin$.

\end{description}

\noindent{\bf Task Scheduling.}
We assume the existence of an entity, called {\em Shared Repository} (whose detailed specification is given in Section~\ref{sec:model}),
that abstracts the service by which clients submit computational tasks to our system and that notifies them when they are completed.This allows our results to be conceptually general, instead of considering specific implementation details.
The Shared Repository is {\em not} a scheduler, since it does not make any task allocation decisions; processors simply access this entity to obtain the set of pending tasks. 
\deleted[ez]{Hence, the \emph{Shared Repository} is a much weaker abstraction than the master processor considered in Master-Worker Internet-based computing systems such as SETI@home~\cite{SETI}, the resource scheduler in computational Grids such as EGEE~\cite{EGEE}, the process scheduler in multicore-based platforms, e.g.,~\cite{AC_RTSS06}, or the decision maker considered in the parallel scheduling literature, e.g.,~\cite{schedulingbook}.}
Such an entity, and implementations of it, have been considered, for example, in the Software Components Communication literature, where it is referred as the Shared Repository Pattern (see for example~\cite{repository1,repository2}, and references therein).  

This makes our setting simpler, easier to implement and more scalable than other popular settings
with stronger scheduling computing entities, such as a {\em central scheduler}. 
Note that even in the case of the central
scheduler, a central repository would still be needed in order for the scheduler to keep track of the pending tasks and proceed
with task allocation. Hence, the underline difference of our setting with that of a central scheduler is that in the latter,
scheduling decisions and processing is done by a single entity which allocates the tasks to the processors, as opposed to our setting where
scheduling decisions are done in {\em parallel} by the participating processors for deciding what task each processor should perform next.
As a consequence, all the results of our work also hold for such stronger models: algorithms work not worse than in the Shared Repository setting since it is a weaker model. The necessary conditions on energy threshold also hold as they are proven for a scenario with a single processor, where these two models are indistinguishable.\vspace{.2em}

\noindent{\bf Related Work.}
The work most closely related to this work is the one by Georgiou and Kowalski~\cite{GeorgiouK11}. As in this work, they consider a task-performing problem where tasks are dynamically and continuously injected to the system, and processors are subject to crashes and restarts. Unlike this work, the computation is broken into synchronous rounds and the notion of {\em per-round} pending-task competitiveness is considered instead.
Furthermore, tasks are assumed to have {\em unit cost}, i.e., they can be performed in one round. The authors consider at first a central scheduler
and then show how and under what conditions it can be implemented in a message-passing distributed setting (called local scheduler). 
They show that even with a central scheduler, no algorithm can be competitive if tasks have different execution times. 
This result has essentially motivated the present work; to use speed-scaling and study the conditions on speedup for which competitiveness is possible. As it turns out, extending the problem for tasks with different processing times
and considering speed-scaling is a non-trivial task; different scheduling policies and techniques had to be devised.

Our work is also related with studies of parallel online scheduling using identical machines~\cite{schedulingbook}. Among them, several papers consider speed-scaling and speedup issues. Some of them,  
unlike our work, consider dynamic scaling (e.g.,~\cite{Albers:2011:MSS:1989493.1989539,Bansal:2009:SSA:1496770.1496846,Chan:2009:SSP:1583991.1583994}).
Usually, in these works preemption is allowed: an execution of a task may be suspended and later restarted from the point of suspension.
In our work, the task must be performed from scratch. The authors of \cite{Greiner:2009:BRS:1583991.1583996} investigate scheduling on $m$ identical speed-scaled processors without migration (tasks are not allowed to move among processors). Among others, they prove that any $z$-competitive online algorithm for a single processor yields a $z B_a$-competitive online algorithm for multiple processors, where $B_a$ is the number of partitions of a set of size $a$.
What is more, unlike our work, the number of processors is not bounded. The work in~\cite{AGM-ICALP11} considers tasks with deadlines (i.e., real-time computing is considered), but no migration, whereas the work in \cite{Albers:2011:MSS:1989493.1989539} considers both. We note that none of these works considers processor failures. Considering failures, as we do, makes parallel scheduling a significantly more challenging problem.\vspace{-1em}

\vspace{.2em}
\section{Model and Definitions\vspace{-.3em}}
\label{sec:model}

\noindent{\bf Computing Setting.}
We consider a system of $n$ homogeneous, fault-prone processors, with unique ids from the set $[n] = \{1,2,\ldots,n\}$.
We assume that processors have access to a shared object, called {\em Shared Repository} or {\em Repository} for short. It 
represents the interface of the system that is used by the clients to submit computational tasks and receive the notifications about the performed ones.\vspace{.2em}

\noindent{\bf Operations.}
The data type of the repository is a set of tasks (to be described later) \replaced[ez]{that}{and it} supports three operations: {\em inject,~get,} and {\em inform}. 
The {\em inject} operation is executed by a client of the system, who adds a task to the current set, and as discussed below, this operation is controlled by an adversary. 
The other two operations are executed by the processors.
By executing a {\em get} operation, a processor obtains from the repository the set of {\em pending tasks}, i.e., the tasks that have been injected into the system, but the repository has not been notified that they have been completed yet.
To simplify the model we assume that, 
if there are no pending tasks when the {\em get} operation is executed, it blocks until some new task is injected, and then it immediately returns the set of new tasks. Upon computing a task, a processor executes an {\em inform} operation, which notifies the repository about the task completion. 
Then the repository removes this task from the set of pending tasks.
Note that due to processor crashes, it would not be helpful for a processor to notify the repository of the task it has scheduled before actually performing the task.
Each operation performed by a processor is associated with a point in time (with the exception of a {\em get} that blocks) and the outcome of the operation is instantaneous (i.e., at the same time point).\vspace{.2em}  

\noindent{\bf Processor cycles.}
Processors run in \emph{real-time cycles}, controlled by an algorithm. Each cycle consists of a {\em get} operation, a computation of a task, and an {\em inform} operation (if a task is completed). Between two consecutive cycles an algorithm may choose to have a processor idling for a period of predefined length.
We assume that the {\em get} and {\em inform} operations consume negligible time (unless {\em get} finds no pending task, in which case it blocks, but returns immediately when a new task is injected). The computation part of the cycle, which 
involves executing a task, consumes the time needed for the specific task to be computed divided by the {\em speedup} $s\ge 1$. 
Processor cycles may not complete: An algorithm may decide to break the current cycle of a processor at any moment, in which case the processor starts a new one. Similarly, a crash failure breaks (forcefully) the cycle of a processor. Then, when the processor restarts, a new cycle begins.\vspace{.2em}
 
\noindent{\bf Work conserving.} We consider all {\em online} algorithms to be {\em work conserving}\replaced[ez]{; not to}{, that is, they never} allow any processor to idle when there are pending tasks and never break a cycle.\vspace{.2em}

\noindent{\bf Event ordering.} 
Due to the concurrent nature of the assumed computing system, processors' cycles may overlap between themselves and with the clients' inject operations. We therefore specify the following event ordering {\em at the repository} at a time $t$: 
first, the {\em inform} operations executed by processors are processed, then the {\em inject} operations, and last the {\em get} operations of processors. This implies that the set of pending tasks returned by a {\em get} operation executed at time $t$ includes, besides the older unperformed tasks, the tasks injected at time $t$, and excludes the tasks reported as performed at time $t$. 
(This event ordering is done only for the ease of presentation and reasoning; it does not affect the generality of results.)
\vspace{.2em}

\noindent{\bf Tasks.}
Each task is associated with a unique {\em identifier}, an {\em arrival time} (the time it was injected in the system based on the repository's clock), and a {\em cost}, measured as the time needed to be performed (without a speedup). 
Let $\lmin$ and $\lmax$ denote the smallest and largest, respectively, costs 
that tasks may have (unless otherwise stated, this information is known to the processors). 
Throughout the paper we refer to a task of cost $\ell\in[\lmin,\lmax]$, as a $\ell$-task.
\remove{
Each task specification $\tau$, is a tuple of the form $(id,~arrival,~cost,~code)$, where $\tau.id$ is a positive integer that uniquely identifies the task in the system, $\tau.arrival$ corresponds to the time at which the task was injected to the system (according to the repository's local clock), $\tau.cost$ is the cost of the task, measured as the time needed for the task to be performed when running without a speedup, and $\tau.code$ corresponds to the computation that needs to occur so that the task is considered completed (that is, the computational part of the task specification that is actually performed).}
We assume that tasks are {\em atomic} with respect to their completion: if a processor stops executing a task 
(intentionally or due to a crash) before completing the entire task, then 
no partial information can be shared with the repository, nor the processor may resume the execution of the task from the point it stopped 
(i.e., preemption is not allowed). Note also, that if a processor performs a task but crashes before the {\em inform} operation, then this task is not considered completed.
Finally, tasks are assumed to be {\em similar} (require equal or comparable resources), {\em independent}, 
and {\em idempotent} (multiple executions of the same task produce the same final result). Several applications involving tasks with such properties are discussed in~\cite{GS_book08}.\vspace{.2em}

\remove{
By {\em similarity}
we mean that the task computations on any processor consume equal
or comparable local resources. By {\em independence}
we mean that the completion of any task does not affect any
other task, and any task can be performed concurrently with any other task.
By {\em idempotence} we mean that each task can be performed one or more times to
produce the same final result.} 

\noindent{\bf Adversary.}
We assume an omniscient adversary that can cause processor crashes and restarts, as well
as task injections (at the repository).
We define an {\em adversarial pattern} $\cA$ as a collection of crash, restart and injection
events caused by the adversary. Each event is associated with the time it occurs (e.g.,
$crash(t,i)$ specifies that processor $i$ is crashed at time $t$).   
%
We say that a processor $i$ is {\em alive} in time interval $[t,t']$, if the processor is
operational at time $t$ and does not crash by time $t'$.
We assume that a restarted 
processor has knowledge of only the algorithm being executed and parameter $n$ (number of processors). 
\replaced[dk]{Thus,}{Therefore,} 
upon a restart, a processor simply starts a new cycle.\vspace{.2em}

\noindent{\bf Efficiency Measures.}
We evaluate our algorithms using the {\em pending cost} measure, defined as follows. Given a time point 
$t\ge 0$ of the execution of an
algorithm $\Alg$ under an adversarial pattern $\cA$, we define the {\em pending cost at time}~$t$, $\cC_t(\Alg,\cA)$,
to be the sum of the costs of the pending tasks 
at the repository at time $t$.
Furthermore, we denote {\em the number of pending tasks} at the repository at time $t$ under adversarial pattern $\cA$ by $\cT_t(\Alg,\cA)$. 
Since we view the task performance problem as an online problem, we pursue competitive analysis.
Specifically, we say that an algorithm $\Alg$ is {\em $x$-pending-cost competitive} if $\cC_t(\Alg,\cA) \leq x\cdot \cC_t(\OPT,\cA) + \Delta$, for any $t$ and under any adversarial pattern $\cA$; $\Delta$ can be any expression independent of $\cA$ and $\cC_t(\OPT,\cA)$ is the minimum (or infimum, in case of infinite computations) pending cost achieved by any {\em off-line algorithm} ---that knows a priori $\cA$ and has unlimited computational power--- at time $t$ of its execution under the adversarial pattern $\cA$. Similarly, we say that an algorithm $\Alg$ is {\em $x$-pending-task competitive} if $\cT_t(\Alg,\cA) \leq x\cdot \cT_t(\OPT,\cA) + \Delta$, where $\cT_t(\OPT,\cA)$ is analogous to $\cC_t(\OPT,\cA)$.
We omit $\cA$ from the above notations when it can be inferred from the context.\vspace{-1em}
   
%
\remove{
Because, as we will show, in the classical competitive measurement described above, no competitiveness can be achieved by deterministic algorithms,
we consider a speed-scaling approach in which the execution of an algorithm is done under
processors speedup $s$, but it is compared to the minimum (infimum) pending cost
achieved by off-line solutions run without speedup; each task takes the amount of time
equal to its cost to be performed by the off-line solution.
We study threshold values of $s$ guarantying competitiveness, and specific competitiveness
formulas are achieved.
\noindent{\bf Global Notation.}
Throughout the paper, we will be using some global notation.
We refer to the cost of a task $\tau$ by $c_{\tau}$. As previously defined, $c_{min}$ and $c_{max}$ are the smallest and largest costs that tasks might have, respectively. The speedup of processors is represented by $s$. 
To refer to a time point in an execution we use $t$ and to denote a time interval $I$. The cost of the tasks injected in an interval $I$ is $\cC_I$ whereas the set of injected tasks in the interval
is $S_I$. \ez{As mentioned already,} the number of pending tasks of an algorithm $\Alg$ at time $t$ is represented by $\cT_t(\Alg)$, whereas the set of pending tasks is $Pending_t(\Alg)$, and their total pending cost is $\cC_t(\Alg)$.
}


\section{NP-hardness}
\label{sec:NPhard}

We now show that the offline problem of optimally scheduling tasks to minimize pending cost or number of pending tasks is NP-hard. This justifies the approach used in this paper for the online problem, speeding up the processors. In fact we show NP-hardness for problems with even one single processor.

Let us consider $\mathit{C\_SCHED}(t,\cA)$ which is the problem of scheduling tasks so that the pending cost at time $t$ under adversarial pattern $\cA$ is minimized. 
We consider a decision version of the problem, $\mathit{DEC\_C\_SCHED}(t,\cA,\omega)$, with an additional input parameter $\omega$. An algorithm solving the decision problem outputs a Boolean value $\mathit{TRUE}$ if and only if there is a schedule that achieves pending cost no more than $\omega$
at time $t$ under adversarial pattern $\cA$. I.e., $\mathit{DEC\_C\_SCHED}(t,\cA,\omega)$ outputs $\mathit{TRUE}$ if and only if $\cC_t(\OPT,\cA) \leq \omega$. The proof of the following theorem can be found in Appendix~\ref{a:NPhard}.

\begin{theorem}
\label{t:nphard}
The problem $\mathit{DEC\_C\_SCHED}(t,\cA,\omega)$ is NP-hard.
\end{theorem}

A similar theorem can be stated (and proved following the same line), 
for a decision version of a respective problem, say $\mathit{DEC\_T\_SCHED}(t,\cA)$ of $\mathit{T\_SCHED}(t,\cA,\omega)$, for which the parameter to be minimized is the number of pending tasks.

\section{Conditions on Non-Competitiveness}
\label{s:non-comp}

For given task costs $\lmin,\lmax$ and speedup $s$, we define parameter $\gamma$ as the smallest number (non-negative integer) of $\lmin$-tasks that one processor can complete in addition to a $\lmax$-task, such that no algorithm running without speedup can complete more tasks in the same time. The following properties are therefore satisfied:

\begin{description}[leftmargin=4mm]
\item [Property 1.] 
$\frac{\gamma\lmin+\lmax}{s} \le (\gamma+1)\lmin$.

\item [Property 2.] 
For every non-negative integer $\kappa < \gamma$, $\frac{\kappa\lmin+\lmax}{s} > (\kappa+1)\lmin$. 
\end{description}

\noindent It is not hard to derive that $\gamma = \max\{\lceil \frac{\lmax-s \lmin}{(s-1)\lmin} \rceil, 0\}$.

\remove{

For given task costs $\lmin,\lmax$ and speedup $s$, 
we define parameter 
$\gamma$ as the smallest non-negative integer that satisfies the following property:

\begin{description}[leftmargin=4mm]
\item [Property 1.] 
$\frac{\gamma\lmin+\lmax}{s} \le (\gamma+1)\lmin$.

It is not hard to derive that $\gamma = \max\{\lceil \frac{\lmax-s \lmin}{(s-1)\lmin} \rceil, 0\}$.
By definition, the following property is also satisfied:

\item [Property 2.] 
For every non-negative integer $\kappa < \gamma$, $\frac{\kappa\lmin+\lmax}{s} > (\kappa+1)\lmin$. 
\end{description}

Intuitively, $\gamma$ gives the smallest number of $\lmin$-tasks that an algorithm \Alg with speedup $s$ can complete in addition to a
$\lmax$-task, such that no algorithm running without speedup can complete more tasks in the same time.

}

We now present and prove necessary conditions for the speedup value to achieve competitiveness. 

\begin{theorem}
\label{t:non-competitive}
For any given $\lmin,\lmax$ and $s$,
if the following two conditions are satisfied

\vspace{1.5mm}
\begin{tabular}{ l r }
  \textbf{(a)} $s < \frac{\lmax}{\lmin}$, and &
  \textbf{(b)} $s < \frac{\gamma \lmin+\lmax}{\lmax}$ \\
\end{tabular}
\vspace{1.5mm}


\noindent then no deterministic algorithm is competitive when run with speedup $s$
against an adversary injecting tasks with cost in $[\lmin,\lmax]$ even in a system with one single processor.
\end{theorem}

In other words, if $s < \min\left\{ \frac{\lmax}{\lmin}, \frac{\gamma \lmin+\lmax}{\lmax}\right\}$ there is no deterministic competitive algorithm.

\begin{proofof}{Theorem~\ref{t:non-competitive}}
Consider a deterministic algorithm \Alg. We define a universal off-line algorithm \OFF with associated crash and injection adversarial patterns, and prove that the cost of \OFF is always bounded while the cost of \Alg is unbounded during the executions of these two algorithms under the defined adversarial crash-injection pattern.

In particular, consider an adversary that activates, and later keeps crashing and re-starting one processor.
The adversarial pattern and the algorithm $\OFF$ are defined recursively in consecutive {\em phases},
where formally each phase is a closed time interval and every two consecutive phases share an end.
In each phase, the processor is restarted in the beginning and crashed at the end of the phase,
while kept continuously alive during the phase.
At the beginning of phase $1$, there are $\gamma$ of $\lmin$-tasks and one $\lmax$-task injected, 
and the processor is activated.

Suppose that we have already defined adversarial pattern and algorithm $\OFF$ till the beginning of phase $i\ge 1$.
Suppose also, that during the execution of $\Alg$ there are $x$ of $\lmin$-tasks and $y$ of $\lmax$-tasks
pending. The adversary does not inject any tasks until the end of the phase.
Under this assumption we could simulate the choices of $\Alg$ during the phase $i$.
There are two cases to consider (illustrated in Figures~\ref{fig:scenario1} and~\ref{fig:scenario2}):
\begin{description}
\item[Scenario 1.]
$\Alg$ schedules $\kappa$ of $\lmin$-tasks, where $0\le \kappa < \gamma$, 
and then schedules a $\lmax$-task;
then $\OFF$ runs $\kappa+1$ of $\lmin$-tasks in the phase, 
and after that the processor is crashed and the phase is finished.
At the end, $\kappa+1$ $\lmin$-tasks are injected.
\item[Scenario 2.] 
$\Alg$ schedules $\kappa=\gamma$ of $\lmin$-tasks;
then $\OFF$ runs a single $\lmax$-task in the phase,
and after that the processor is crashed and the phase is finished.
At the end, one $\lmax$-task is injected.
\end{description}

What remains to show is that the definitions of the $\OFF$ algorithm and the associated adversarial pattern are valid, and that in the execution of $\OFF$ the number of pending tasks is bounded, while in the corresponding
execution of $\Alg$ it is not bounded. Since the tasks have bounded cost, the same applies to the pending cost of both $\OFF$ and $\Alg$.
Here we give some useful properties of the considered executions of algorithms $\Alg$ and $\OFF$, whose proofs can be found in Appendix~\ref{app:NON-COMP}.

\begin{lemma}
\label{l:off-pending}
The phases, adversarial pattern and algorithm $\OFF$ are well-defined.
Moreover, in the beginning of each phase, there are exactly $\gamma$ of $\lmin$-tasks and 
one $\lmax$-task pending in the execution of~$\OFF$.
\end{lemma}


\begin{lemma}
\label{l:infinite}
There are infinite number of phases.
\end{lemma}


\begin{lemma}
\label{l:no-long}
$\Alg$ never performs any $\lmax$-task.
\end{lemma}


\begin{lemma}
\label{l:opt-long}
If Scenario~2 was applied in the specification of a phase $i$, then the number of pending $\lmax$-tasks at the end of phase $i$ in the execution of $\Alg$ increases by one comparing with the beginning of phase $i$, while the number of pending $\lmax$-tasks stays the same in the execution of $\OFF$.\vspace{-.7em}
\end{lemma}


Now we resume the main proof of non competitiveness, i.e., Theorem~\ref{t:non-competitive}.
By Lemma~\ref{l:off-pending}, the adversarial pattern and the corresponding offline algorithm $\OFF$
are well-defined and by Lemma~\ref{l:infinite}, the number of phases is infinite.
There are therefore two cases to consider:
(1) If the number of phases for which Scenario~2 was applied in the definition is infinite, then 
by Lemma~\ref{l:opt-long} the number of pending $\lmax$-tasks increases by one infinitely many times, while by Lemma~\ref{l:no-long} it never decreases. Hence it is unbounded.
(2) Otherwise (i.e., if the number of phases for which Scenario~2 was applied in the definition is bounded),
after the last Scenario~2 phase in the execution of $\Alg$, there are only phases in which Scenario~1 is
applied, and there are infinitely many of them. 
In each such phase, $\Alg$ performs only $\kappa$ of  $\lmin$-tasks while 
$\kappa+1$ $\lmin$-tasks will be injected at the end of the phase, for some corresponding non-negative
integer $\kappa<\gamma$ defined in the specification of Scenario~1 for this phase.
Indeed, the length of the phase is $(\kappa+1)\lmin$, while after performing $\kappa$ of  $\lmin$-tasks
$\Alg$ schedules a  $\lmax$-task and the processor is crashed before completing it, because 
$\frac{\kappa\lmin+\lmax}{s} > (\kappa+1)\lmin$ (cf., Property 2).
Therefore, in every such phase of the execution of $\Alg$ the number of pending $\lmin$-tasks increases
by one, and it does not decrease since there are no other kinds of phases (recall that we consider
phases with Scenario~1 after the last phase with Scenario~2 finished). Hence the number
of $\lmin$-tasks grows unboundedly in the execution of $\Alg$.

To conclude, in both cases above, the number of pending tasks in the execution of $\Alg$ grows unboundedly in time,
while the number of pending tasks in the corresponding execution of $\OFF$ (for the same adversarial pattern)
is always bounded, by Lemma~\ref{l:off-pending}.
\end{proofof}

Note that the use of condition (a) is implicit in our proof. 
\vspace{-1em}


\section{Algorithm $\LISs$}
\label{s:alg}

In this section we present Algorithm $\LISs$,  
which balances between the following two paradigms: scheduling Longest-In-System task first (LIS) and redundancy avoidance. More precisely, the algorithm at a processor tries to schedule the task that has been waiting the longest and does not cause redundancy of work if the number of pending tasks is sufficiently large. See the
algorithm pseudocode for details. 

\vspace{.5em}
\begin{footnotesize}
\hrule \vspace{2pt}
\noindent {\bf Algorithm $\LISs$ (for processor $p$)}\vspace{1pt}
\hrule \vspace{2pt}
\noindent {\bf Repeat}\hfill {\footnotesize\em //Upon awaking or restart, start here}\\
\indent {\bf Get} from the Repository the set of pending tasks $Pending$;\\
\indent {\bf Sort} $Pending$ 
by task arrival and ids/costs;\\
\indent {\bf If} $|Pending| \geq 1$ \\
\indent\indent {\bf then} perform task with rank $p \cdot \beta n\mod |Pending|$;\\
\indent {\bf Inform} the Repository of the task performed.   
\hrule
\end{footnotesize}

\vspace{1em}

\noindent 
Observe that since $s \geq \lmax/\lmin$, Algorithm $\LISs$ is able to complete one task for each task completed by the offline algorithm. Additionally, if there are at least $\beta n^2$ tasks pending, for $\beta\geq\frac{\lmax}{\lmin}$, two processors do not schedule the same task. Combining these two observations it is possible to prove that $\LISs$ is $1$-task-competitive. 


\begin{theorem}
\label{t:LISs-comp1}
$\cT_t(\LISs,\cA) \leq \cT_t(\OPT,\cA) + \beta n^2+3n$ 
and $\cC_t(\LISs,\cA) \leq \frac{\lmax}{\lmin}\cdot\left(\cC_t(\OPT,\cA) + \beta n^2+3n\right)$, for any time $t$ and adversarial pattern $\cA$, and for speedup $s\geq\frac{\lmax}{\lmin}$, when 
$\beta\geq\frac{\lmax}{\lmin}$.\vspace{-.3em}
\end{theorem}




\begin{proof}
\deleted[af]{Before giving the proof of Theorem~\ref{t:LISs-comp1}, we analyze several properties of the algorithm needed for this proof.}
We first focus on the number of pending-tasks.
Suppose that $\LISs$ is not $\OPT + \beta n^2+3n$ competitive in terms of the number of pending tasks, $\OPT$, for some $\beta \geq \frac{\lmax}{\lmin}$
and some $s\ge \frac{\lmax}{\lmin}$. Consider an execution witnessing this fact and fix the adversarial pattern associated with it together with the optimum solution \OPT\ for it.

Let $t^*$ be a time in the execution when $\cT_{t^*}(\LISs) > \cT_{t^*}(\OPT) + \beta n^2+3n$.
For any time interval $I$, let $\cT_I$ be the total number of tasks injected in the interval $I$. Let $t_*\le t^*$ be the smallest time such that for all $t\in [t_*,t^*)$, $\cT_t(\LISs) > \cT_t(\OPT) +  \beta n^2$
(Note that the selection of minimum time satisfying some properties defined by the computation is possible due to the fact that the computation is split into discrete processor cycles.)
Observe that $\cT_{t_*}(\LISs) \le \cT_{t_*}(\OPT) + \beta n^2 +n$, because at time $t_*$ no more than $n$ tasks could be reported to the repository by \OPT, while just before $t_*$ the difference between $\LISs$ and \OPT was at most $\beta n^2$.


\added[af]{Then, we have the following property, whose proof is given in Appendix~\ref{app:LIS}.}

\begin{lemma}
\label{l:number-tasks}
$\cT_{t^*}(\LISs) \le \cT_{t^*}(\OPT) + \beta n^2+3n$.
\end{lemma}

\remove{

\begin{proof}
By Lemma~\ref{l:interval} we have that $\cT_{t_*+\lmin}(\LISs) \le \cT_{t_*+\lmin}(\OPT) + \beta n^2+2n$.\\
Let $y$ be the total number of tasks reported by $\LISs$ in $(t_*+\lmin,t^*]$.
By Lemma~\ref{l:OPT-not-better} and definitions $t_*$ and $t^*$, $\OPT$ reports no more that $y+n$ tasks in $(t_*+\lmin,t^*]$. Therefore,
\[
\cT_{t^*}(\OPT)
\ge
\cT_{t_*+\lmin}(\OPT)-(y+n)
\ .
\]
By Lemma~\ref{l:redundancy}, in the interval $(t_*+\lmin,t^*]$, no redundant work is reported by $\LISs$. Thus,
\[
\cT_{t^*}(\LISs)
\le
\cT_{t_*+\lmin}(\LISs)-y
\ .
\]

Consequently,
\begin{eqnarray*}
\cT_{t^*}(\LISs)
&\le&
\cT_{t_*+\lmin}(\LISs) - y
\\
&\le&
\left(\cT_{t_*+\lmin}(\OPT)+\beta n^2+2n\right) - y
\\
&\le&
\cT_{t^*}(\OPT)+(\beta n^2 +2n) + n
\\
&\le&
\cT_{t^*}(\OPT)+\beta n^2 + 3n
\end{eqnarray*}
as desired.\vspace{.3em}
\end{proof}

} 

\deleted[af]{We are now ready to show the proof of Theorem~\ref{t:LISs-comp1}.\vspace{.2em}}

\deleted[af]{\noindent{\em Proof of Theorem ~\ref{t:LISs-comp1}:}}
The competitiveness for the number of pending tasks follows directly from Lemma~\ref{l:number-tasks}: it violates the contradictory assumptions made in the beginning of the analysis.
The result for the pending cost is a direct consequence of the one for pending tasks, as the cost of any pending task in $\LISs$ is at most $\frac{\lmax}{\lmin}$ times bigger than the cost of any pending task in $\OPT$.\hfill\rule{2mm}{2mm}
\end{proof}




\section{Algorithm $\gnBurst$\vspace{-.5em}}
\label{sec:gnBurst}

Observe that, against an adversarial strategy where at first only one $\lmax$-task is injected, and then only $\lmin$-tasks are injected, algorithm $\LISs$ with one processor has unbounded competitiveness when $s < \frac{\lmax}{\lmin}$
(this can be generalized for $n$ processors). This is also the 
case for algorithms using many other scheduling policies, e.g., ones that schedule first the more costly tasks.
This suggests that for $s < \frac{\lmax}{\lmin}$ a scheduling policy that alternates executions of 
lower-cost and higher-cost tasks should be devised.
In this section, we show that if the speed-up satisfies $\frac{\gamma \lmin+\lmax}{\lmax} \leq s < \frac{\lmax}{\lmin}$ and the tasks can have only two different costs, $\lmin$ and $\lmax$, then there is an algorithm, call it $\gnBurst$, that achieves 1-pending-task and 1-pending-cost competitiveness in a system with $n$ processors.
\deleted[ez]{Recall that $\gamma = \max\{\lceil \frac{\lmax - s \lmin}{(s-1)\lmin} \rceil, 0\}$. It is easy to verify that if $s < \frac{\lmax}{\lmin}$ then $\gamma = \lceil \frac{\lmax - s \lmin}{(s-1)\lmin} \rceil > 0$. This implies that $s>1$.}
The algorithm's pseudocode follows.

\vspace{.3em}
\begin{footnotesize}
\hrule \vspace{2pt} 
\noindent {\bf Algorithm $\gnBurst$ (for processor $p$)} \vspace{1pt}
\hrule \vspace{2pt}
\noindent {\bf Input:} $\lmin, \lmax, n, s$ \\
\noindent {\bf Calculate} $\gamma \gets \lceil \frac{\lmax - s\lmin}{(s-1)\lmin} \rceil$ \\
\noindent {\bf Repeat}\hfill {\footnotesize\em //Upon awaking or restart, start here}\\
\indent $c \gets 0$; \hfill {\footnotesize\em //Reset the counter 
}\\
\indent {\bf Get} from the Repository the set of pending tasks $Pending$;\\
\indent {\bf Create} lists $L_{min}$ and $L_{max}$ of
$\lmin$- and $\lmax$-tasks; \\ 
\indent {\bf Sort} $L_{min}$ and $L_{max}$ according to task arrival;\\
\indent {\bf Case 1:} $|L_{min}| < n^2$ and $|L_{max}| < n^2$ \\
\indent \indent {\bf If} previously performed task was of cost $\lmin$ {\bf then} \\
\indent \indent \indent perform task $(p\cdot n) \mod |L_{max}|$ in $L_{max}$; 
\: $c \gets 0$; \hfill {\footnotesize\em //Reset the counter 
}\\
\indent \indent {\bf else}  perform task $(p\cdot n) \mod |L_{min}|$ in $L_{min}$;
\: $c \gets \min(c+1,\gamma)$;\\
\indent {\bf Case 2:} $|L_{min}| \geq n^2$ and $|L_{max}| < n^2$ \\
\indent \indent perform the task at position $p\cdot n$ in $L_{min}$; 
\: $c \gets \min(c+1,\gamma)$; \\
\indent {\bf Case 3:} $|L_{min}| < n^2$ and $|L_{max}| \geq n^2$ \\
\indent \indent perform the task at position $p\cdot n$ in $L_{max}$; 
\: $c \gets 0$; \hfill {\footnotesize\em //Reset the counter 
}\\
\indent {\bf Case 4:} $|L_{min}| \geq n^2$ and $|L_{max}| \geq n^2$ \\
\indent \indent {\bf If} $c$ = $\gamma$ {\bf then} perform task at position $p\cdot n$ in $L_{max}$; 
\: $c \gets 0$; \hfill {\footnotesize\em //Reset the counter 
}\\
\indent \indent {\bf else} perform task at position $p\cdot n$ in $L_{min}$; 
\: $c \gets \min(c+1,\gamma)$; \\
\indent {\bf Inform} the Repository of the task performed.
\hrule
\end{footnotesize}
\vspace{1em}

We first overview the main idea behind the algorithm. Each processor groups the set of pending tasks into two sublists, 
$L_{min}$ and $L_{max}$, each corresponding to the tasks of cost $\lmin$ and $\lmax$, respectively, ordered by arrival time. 
Following the same idea behind Algorithm $\LISs$, the algorithm avoids redundancy when ``enough'' tasks are pending. 
Furthermore, the algorithm needs to take into consideration parameter $\gamma$ and the bounds on speed-up $s$.
For example, in the case that there exist enough $\lmin$- and $\lmax$-tasks (more than $n^2$ to be exact) each
processor performs no more than $\gamma$ consecutive $\lmin$-tasks and then performs a $\lmax$-task; this is the time
it takes for the same processor to perform a $\lmax$-task in \OPT. 
To this respect, a counter is used to keep track of the number of consecutive $\lmin$-tasks, which is {\em reset} when a 
$\lmax$-task is performed. Special care needs to be taken for all other cases, e.g., when there are more 
than $n^2$ $\lmax$-tasks pending but less than $\lmin$-tasks, etc.

The analysis of $\gnBurst$ proving the following bound for both $\lmax$- and $\lmin$-tasks is in Appendix~\ref{a:gnburst}.\vspace{-.7em}

\begin{theorem}
\label{t:gamma-b-com}
$\cT_t(\gnBurst,\cA) \leq \cT_t(\OPT,\cA) + 2n^2 + (3+\left\lceil \frac{\lmax}{s  \cdot \lmin} \right\rceil)n,$ for any time $t$ and adversarial pattern $\cA$.\vspace{-.7em} 
\end{theorem}

The difference in the number of $\lmax$-tasks between $\Alg$ and $\OPT$ can be bounded by $n^2+2n$ (see Lemma~\ref{l-ytasks}). This,
and Theorem~\ref{t:gamma-b-com}, yield the following bound on the pending cost of $\gnBurst$, which also implies that it is 1-pending-cost competitive.\vspace{-.5em}

\begin{theorem}
$\cC_t(\gnBurst,\cA) \leq \cC_t(\OPT,\cA) +  \lmax(n^2 + 2n) + \lmin(n^2 + (1+\left\lceil \frac{\lmax}{s  \cdot \lmin} \right\rceil)n)$, for any time $t$ and adversarial pattern $\cA$.\vspace{-.8em}
\end{theorem}



\section{Algorithm $\LAF$\vspace{-.5em}}
\label{sec:LAF}
In the case of only two different costs, we can obtain a competitive solution for speedup that
matches the lower bound from Theorem~\ref{t:non-competitive}. More precisely, for given two 
different cost values, $\lmin$ and $\lmax$,
we can compute the minimum speedup $s^*$ satisfying condition (b) from Theorem~\ref{t:non-competitive} 
for these two costs,
and choose $\LISs$ with speedup $\lmax/\lmin$ in case $\lmax/\lmin\le s^*$ and $\gnBurst$ 
with speedup $s^*$ otherwise\footnote{Note that $s^*$ is upper bounded by $2$, as explained in Appendix~\ref{app:conditions}.}.
However, in the case of more than two different task costs we cannot use $\gnBurst$, and so far we could only
rely on $\LISs$ with speedup $\lmax/\lmin$, which can be large.

We would like to design a ``substitute'' for algorithm $\gnBurst$, working for any bounded
number of different task costs, which is competitive for some fixed small speedup.  
(Note that $s\geq 2$ is enough to guarantee that condition (b) does not hold.)
This algorithm would be used when $\lmax/\lmin$ is large.
In this section we design such an algorithm, that works for any bounded
number of different task costs, and is competitive for speedup
$s\ge 7/2$.
This algorithm, together with algorithm $\LISs$, guarantee competitiveness
for speedup 
$s\ge \min\{\frac{\lmax}{\lmin},7/2\}$.
In more detail, one could apply $\LISs$ with speedup $\frac{\lmax}{\lmin}$ when
$\frac{\lmax}{\lmin}\le 7/2$ 
and the new algorithm with speedup
$7/2$
otherwise.

We call the new algorithm {\em Largest\_Amortized\_Fit} or $\LAF$ for short.
It is parametrized by $\beta\ge \lmax/\lmin$. This algorithm is more ``geared'' towards pending cost efficiency.
In particular, each processor keeps the variable $\total$, storing the total cost of tasks reported by processor $p$, since the last restart
(recall that upon a restart processors have no recollection of the past).
For every possible task cost, pending tasks of that cost are sorted using the Longest-in-System (LIS) policy.
Each processor schedules the largest cost task which is not bigger than $\total$ and is such, that 
the list of pending tasks of the same cost (as the one selected) has at least $\beta n^2$ elements,
for $\beta\ge \lmax/\lmin$. If there is no such task then the processor schedules an arbitrary pending one.

As we prove in Appendix~\ref{app:LAF}, in order for the algorithm to be competitive,
the number of different costs of injected tasks must be finite in the  range $[\lmin,\lmax]$.
Otherwise, the number of tasks of the same cost might never be larger than $\beta n^2$, which is
necessary to assure redundancy avoidance. Whenever this redundancy avoidance is possible, the
algorithm behaves in a conservative way in the sense that it schedules a large task, but not larger than
the total cost already completed. This implies that in every life period of a processor (the continuous period between a restart
and a crash of the processor) only a constant fraction of this period could be wasted (wrt the total task cost covered by \OPT
in the same period). Based on this observation, a non-trivial argument 
shows  that a constant speedup suffices for obtaining 1-pending-cost competitiveness.\vspace{-.5em}

%

\begin{theorem}
\label{thm:LAF}
Algorithm $\LAF$ is 1-pending-cost competitive, and thus $\frac{\lmax}{\lmin}$-pending-task competitive, for speedup 
$s\ge 7/2$, 
provided the number of different costs of tasks in the execution is finite.
\vspace{-.5em}
\end{theorem}

\section{Conclusions\vspace{-.5em}}
\label{sec:conclusions}

In this paper we have shown that a speedup $s \geq \min\left\{ \frac{\lmax}{\lmin}, \frac{\gamma \lmin+\lmax}{\lmax}\right\}$ is {\em necessary} and {\em sufficient} for competitiveness. 

\deleted[ez]{
Additionally, in Appendix~\ref{sec:inac} we take an inside view of the case when the system does not provide the exact values of task costs, but instead it may provide only an estimate of the cost of each task. We refer to this setting as $\varepsilon$-inaccurate and study the details of the modifications that must take place.}

One could argue that the algorithms we propose assume the knowledge of $\lmin$ and $\lmax$, 
which may seem unrealistic. However, in practice, processors can estimate the smallest and largest task costs from the costs seen so far, 
and use these values as $\lmin$ and $\lmax$ in the algorithms.
This results in a similar performance (up to constant factors) of the proposed algorithms with this adaptive computation of $\lmin$ and $\lmax$ 
with some minor changes in the analysis.
%

A research line that we believe worth of further investigation is to study systems where processors could use different speedups or their speedup could vary over time\replaced[ez]{ or even}{. Another interesting direction is} to accommodate 
dependent tasks.


\bibliographystyle{plain}
\bibliography{Bibliography,references}

\begin{thebibliography}{10}

\bibitem{Intelwhitepaper}
Enhanced intel speedstep technology for the intel pentium m processor.
\newblock Intel White Paper 301170-001, 2004.

\bibitem{AADW_FOCS94}
M.~Ajtai, J.~Aspnes, C.~Dwork, and O.~Waarts.
\newblock A theory of competitive analysis for distributed algorithms.
\newblock In {\em Proceedings of the 35th Symposium on Foundations of Computer
  Science (FOCS 1994)}, pages 401--411, 1994.

\bibitem{Albers:2012:RIN:2095116.2095216}
Susanne Albers and Antonios Antoniadis.
\newblock Race to idle: New algorithms for speed scaling with a sleep state.
\newblock In {\em Proceedings of the 23rd ACM-SIAM Symposium on Discrete
  Algorithms (SODA 2012}, pages 1266--1285, 2012.

\bibitem{Albers:2011:MSS:1989493.1989539}
Susanne Albers, Antonios Antoniadis, and Gero Greiner.
\newblock On multi-processor speed scaling with migration.
\newblock In {\em Proceedings of the 23rd ACM Symposium on Parallelism in
  Algorithms and Architectures (SPAA 2011)}, pages 279--288, 2011.

\bibitem{AllocateFOCS}
Dan Alistarh, Michael~A. Bender, Seth Gilbert, and Rachid Guerraoui.
\newblock How to allocate tasks asynchronously.
\newblock In {\em Proceedings of the 53rd IEEE Symposium on Foundations of
  Computer Science (FOCS 2012)}, pages 331--340, 2012.

\bibitem{AGM-ICALP11}
S.~Anand, Naveen Garg, and Nicole Megow.
\newblock Meeting deadlines: How much speed suffices?
\newblock In {\em Proceedings of the 38th International Colloquium on Automata,
  Languages and Programming (ICALP 2011)}, pages 232--243, 2011.

\bibitem{AW_SIAM97}
R.J. Anderson and H.~Woll.
\newblock Algorithms for the certified {W}rite-{A}ll problem.
\newblock {\em SIAM Journal of Computing}, 26(5):1277--1283, 1997.

\bibitem{AKP_STOC92}
B.~Awerbuch, S.~Kutten, and D.~Peleg.
\newblock Competitive distributed job scheduling.
\newblock In {\em Proceedings of the 24th {ACM} Symposium on Theory of
  Computing (STOC 1992)}, pages 571--580, 1992.

\bibitem{Bansal:2009:SSA:1496770.1496846}
Nikhil Bansal, Ho-Leung Chan, and Kirk Pruhs.
\newblock Speed scaling with an arbitrary power function.
\newblock In {\em Proceedings of the 20th ACM-SIAM Symposium on Discrete
  Algorithms (SODA 2009)}, pages 693--701, 2009.

\bibitem{Chan:2009:SSP:1583991.1583994}
Ho~Leung Chan, Jeff Edmonds, and Kirk Pruhs.
\newblock Speed scaling of processes with arbitrary speedup curves on a
  multiprocessor.
\newblock In {\em Proceedings of the 21st ACM Symposium on Parallelism in
  Algorithms and Architectures (SPAA 2009)}, pages 1--10, 2009.

\bibitem{CDS_DC01}
B.~Chlebus, R.~De-Prisco, and A.A. Shvartsman.
\newblock Performing tasks on restartable message-passing processors.
\newblock {\em Distributed Computing}, 14(1):49--64, 2001.

\bibitem{CMR_TPDS07}
G.~Cordasco, G.~Malewicz, and A.~Rosenberg.
\newblock Advances in {IC-S}cheduling theory: Scheduling expansive and
  reductive dags and scheduling dags via duality.
\newblock {\em IEEE Transactions on Parallel and Distributed Systems},
  18(11):1607--1617, 2007.

\bibitem{DOOPM_MTAGS10}
J.~Dias, E.~Ogasawara, D.~de~Oliveira, E.~Pacitti, and M.~Mattoso.
\newblock A lightweight execution framework for massive independent tasks.
\newblock In {\em Proceedings of the 3rd {IEEE} Workshop on Many-Task Computing
  on Grids and Supercomputers}, 2010.

\bibitem{Emeketal_PODC10}
Y.~Emek, M.~M. Halldorsson, Y.~Mansour, B.~Patt-Shamir, J.~Radhakrishnan, and
  D.~Rawitz.
\newblock Online set packing and competitive scheduling of multi-part tasks.
\newblock In {\em Proceedings of the 29th {ACM} Symposium on Principles of
  Distributed Computing (PODC 2010)}, page 2010, 440--449.

\bibitem{EGEE}
Enabling~Grids for E-sciencE (EGEE).
\newblock \url{http://www.eu-egee.org}.

\bibitem{GeorgiouK11}
Chryssis Georgiou and Dariusz~R. Kowalski.
\newblock Performing dynamically injected tasks on processes prone to crashes
  and restarts.
\newblock In {\em Proceedings of the 25th International Symposium on
  Distributed Computing, (DISC 2011)}, pages 165--180. Springer, 2011.

\bibitem{GS_book08}
Chryssis Georgiou and Alexander~A. Shvartsman.
\newblock {\em Do-All Computing in Distributed Systems: Cooperation in the
  Presence of Adversity}.
\newblock Springer, 2008.

\bibitem{Greiner:2009:BRS:1583991.1583996}
Gero Greiner, Tim Nonner, and Alexander Souza.
\newblock The bell is ringing in speed-scaled multiprocessor scheduling.
\newblock In {\em Proceedings of the 21st ACM Symposium on Parallelism in
  Algorithms and Architectures (SPAA 2009)}, pages 11--18, 2009.

\bibitem{166609}
K.S. Hong and J.Y.-T. Leung.
\newblock On-line scheduling of real-time tasks.
\newblock {\em IEEE Transactions on Computers}, 41(10):1326--1331, 1992.

\bibitem{160366}
K.~Jeffay, D.F. Stanat, and C.U. Martel.
\newblock On non-preemptive scheduling of period and sporadic tasks.
\newblock In {\em Proceedings of the 12th Real-Time Systems Symposium}, pages
  129--139, 1991.

\bibitem{KS97}
P.C. Kanellakis and A.A. Shvartsman.
\newblock {\em Fault-Tolerant Parallel Computation}.
\newblock Kluwer Academic Publishers, 1997.

\bibitem{SETI}
E.~Korpela, D.~Werthimer, D.~Anderson, J.~Cobb, and M.~Lebofsky.
\newblock Seti@home: Massively distributed computing for seti.
\newblock {\em Computing in Science and Engineering}, 3(1):78--83, 2001.

\bibitem{repository1}
Philippe Lalanda.
\newblock Shared repository pattern.
\newblock In {\em Proceedings of the 5th Pattern Languages of Programs
  Conference (PLoP 1998)}, 1998.

\bibitem{DBLP:journals/algorithmica/PhillipsSTW02}
Cynthia~A. Phillips, Clifford Stein, Eric Torng, and Joel Wein.
\newblock Optimal time-critical scheduling via resource augmentation.
\newblock {\em Algorithmica}, 32(2):163--200, 2002.

\bibitem{schedulingbook}
Michael~L. Pinedo.
\newblock {\em Scheduling: Theory, Algorithms, and Systems}.
\newblock Springer, fourth edition, 2012.

\bibitem{DBLP:journals/tse/SchwanZ92}
Karsten Schwan and Hongyi Zhou.
\newblock Dynamic scheduling of hard real-time tasks and real-time threads.
\newblock {\em IEEE Trans. Software Eng.}, 18(8):736--748, 1992.

\bibitem{ST_CACM85}
D.~Sleator and R.~Tarjan.
\newblock Amortized efficiency of list update and paging rules.
\newblock {\em Communications of the ACM}, 28(2):202--208, 1985.

\bibitem{repository2}
Uwe van Heesch, Sara~Mahdavi Hezavehi, and Paris Avgeriou.
\newblock Combining architectural patterns and software technologies in one
  design language.
\newblock In {\em Proceedings of the 16th European Pattern Languages of
  Programming (EuroPLoP 2011)}, 2011.

\bibitem{5062123}
A.~Wierman, L.L.H. Andrew, and Ao~Tang.
\newblock Power-aware speed scaling in processor sharing systems.
\newblock In {\em Proceedings of IEEE INFOCOM 2009}, pages 2007--2015, 2009.

\bibitem{DBLP:conf/focs/YaoDS95}
F.~Frances Yao, Alan~J. Demers, and Scott Shenker.
\newblock A scheduling model for reduced {CPU} energy.
\newblock In {\em Proceedings of the 36th IEEE Symposium on Foundations of
  Computer Science (FOCS 1995)}, pages 374--382, 1995.

\bibitem{KComputer}
M.~Yokokawa, F.~Shoji, A.~Uno, M.~Kurokawa, and T.~Watanabe.
\newblock The k computer: Japanese next-generation supercomputer development
  project.
\newblock In {\em Proceedings of the 2011 International Symposium on Low Power
  Electronics and Design (ISLPED 2011)}, pages 371--372, 2011.

\end{thebibliography}

\remove{

}

\newpage

\section*{Figures}

\begin{figure*}[h]
\hspace*{-1em}
\centering
\includegraphics[width=6.7in]{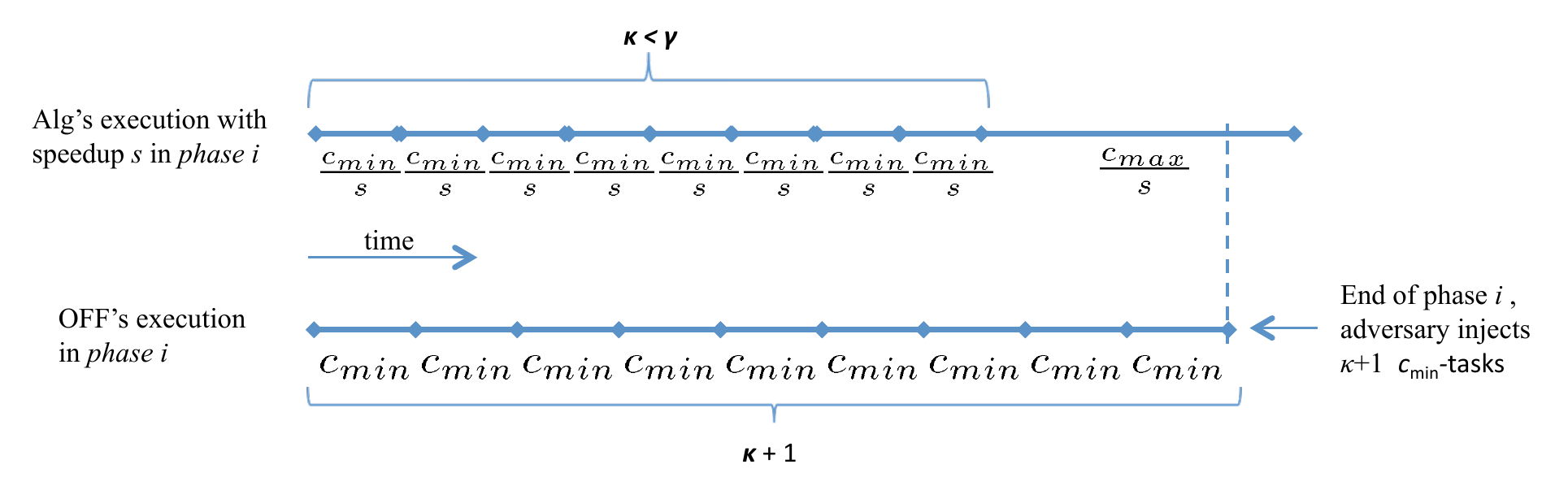}
\caption{\small Illustration of  Scenario 1. It uses the property $(\kappa\lmin+\lmax)/s > (\kappa+1)\lmin$, for any integer $0 \le \kappa < \gamma$ (Property 2).}
\label{fig:scenario1}
\end{figure*}

\begin{figure*}[h]
\hspace*{-1em}
\centering
\includegraphics[width=6.7in]{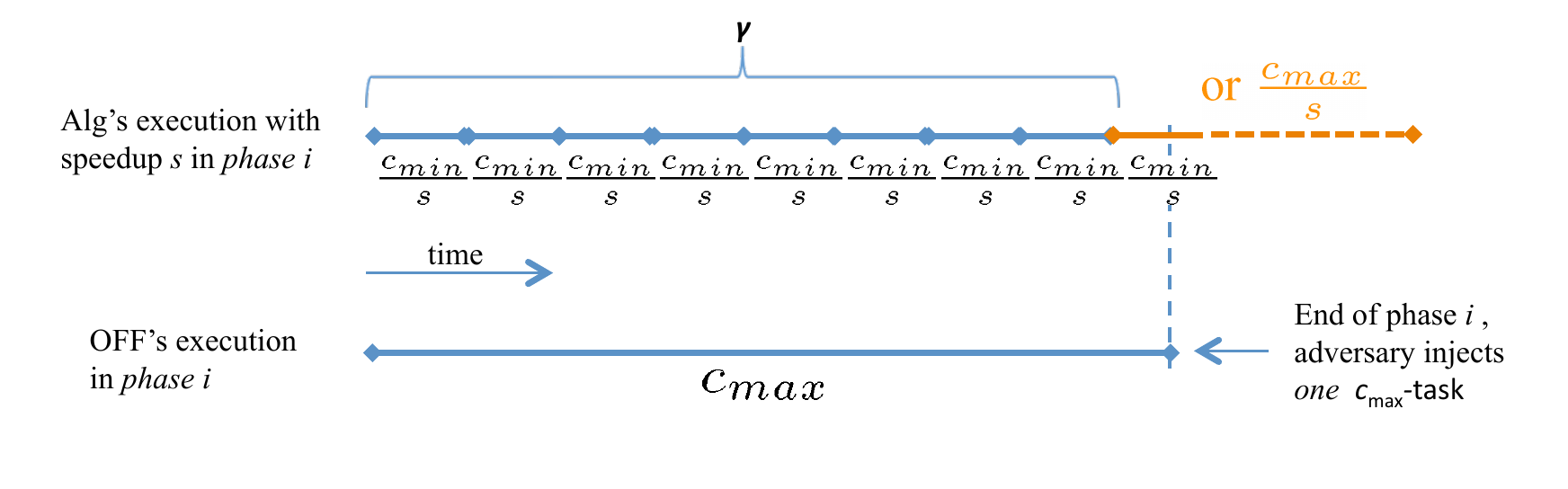}
\caption{\small Illustration of  Scenario 2. It uses the property $(\gamma\lmin+\lmax)/s > \lmax$ (condition (b) of Theorem~\ref{t:non-competitive}).}
\label{fig:scenario2}
\end{figure*}

\newpage

\appendix

\section*{APPENDIX}


\section{\replaced[af]{Omitted}{Missing} proofs from Section \ref{sec:NPhard}}
\label{a:NPhard}

\begin{proofof}{Theorem~\ref{t:nphard}}
The reduction we use is from the Partition problem. The input considered is a set of numbers (we assume positive) $C=\{x_1, x_2, ..., x_k\}$, $k >1$. The problem is to decide whether there is a subset $C' \subset C$ such that $\sum_{x_i \in C'} x_i=\frac{1}{2}\sum_{x_i \in C} x_i$.
The Partition problem is know to be NP-complete.

Consider any instance $I_p$ of Partition. We construct an instance $I_d$ of $\mathit{DEC\_C\_SCHED}(t,\cA,\omega)$ as follows.
The time $t$ is set to $1 + \sum_{x_i \in C} x_i$. 
The adversarial pattern $\cA$ injects a set $S$ of $k$ tasks at time $0$, so that the $i$th task has cost $x_i$. It also starts the processor at time 0 and crashes it at time $\frac{1}{2}\sum_{x_i \in C} x_i$. Then,
$\cA$ restarts the processor immediately and crashes it again at time $\sum_{x_i \in C} x_i$. The processor does not 
restart until time $t$. Finally, the parameter $\omega$ is set to $0$.

Assume there is an algorithm $\Alg$ that solves $\mathit{DEC\_C\_SCHED}$. We show that $\Alg$ can be used
to solve the instance $I_p$ of Partition by solving the instance $I_d$ of $\mathit{DEC\_C\_SCHED}$ obtained as described.
If there is a $C' \subset C$ such that $\sum_{x_i \in C'} x_i=\frac{1}{2}\sum_{x_i \in C} x_i$, then there is an algorithm that is able to
schedule tasks from $S$ so that the two semi-periods (of length $\frac{1}{2}\sum_{x_i \in C} x_i$ each) the processor is active, it is doing useful work.
In that case, the pending cost at time $t$ will be $0=\omega$. If, on the other hand,
such subset does not exist, some of the time the processor is active will be wasted, and the
cost pending at time $t$ has to be larger than $\omega$.
\end{proofof}

\section{\replaced[af]{Omitted}{Missing} proofs from Section \ref{s:non-comp}}
\label{app:NON-COMP}

\begin{proofof}{Lemma~\ref{l:off-pending}}
We argue by induction on the number of phases that: 
in the beginning of phase $i$ there are exactly $\gamma$ of $\lmin$-tasks and one $\lmax$-task pending in the execution of $\OFF$, and therefore phase $i$ is well defined. Its specification (including termination time) depends only on whether $\OFF$ schedules either $\gamma$ of $\lmin$-tasks (in Scenario 1) or one $\lmax$-task (in Scenario 2) before the next task injection at the end of the phase.
The invariant holds for phase $1$ by definition. By straightforward investigation of both Scenarios, the very same configuration of task lengths that has been performed by $\OFF$ in its execution during a phase is injected at the end of the phase, and therefore the inductive argument proves the invariant for every consecutive phase.
\end{proofof}

\begin{proofof}{Lemma~\ref{l:infinite}}
First, by Lemma~\ref{l:off-pending}, consecutive phases are well-defined.
Second, observe
that each phase is finite, regardless of whether Scenario~1 or Scenario~2 is applied,
as it is bounded by the time in which $\OFF$ performs either at most $\gamma$ of $\lmin$-tasks
(in Scenario~1) or one $\lmax$-task (in Scenario~2).
Hence, the number of phases is infinite.\vspace{.2em}
\end{proofof}

\begin{proofof}{Lemma~\ref{l:no-long}}
It follows from the specification of Scenarios~1 and 2, condition (b) on the speedup $s$, 
and from Property 2.
Consider a phase. If Scenario~1 is applied for specification of this phase then
$\Alg$ could not finish its $\lmax$-task scheduled after $\kappa<\gamma$
$\lmin$-tasks, because the time needed for completing this sequence of tasks is at least
$\frac{\kappa\lmin+\lmax}{s}$, which, by Property 2, is larger than the length of this phase $(\kappa+1)\lmin$.
If Scenario~2 is applied for specification of this phase, then the first $\lmax$-task
could be finished by $\Alg$ no earlier than $\frac{\gamma\lmin+\lmax}{s}$ time
after the beginning of the phase, which is again bigger than the length of this phase $\lmax$,
by the assumption (b) on the speedup $s< \frac{\gamma \lmin+\lmax}{\lmax}$.
\end{proofof}

\begin{proofof}{Lemma~\ref{l:opt-long}}
It follows from Lemma~\ref{l:no-long} and from specification of tasks injections at the end of phase $i$, by Scenario~2.
\end{proofof}

\section{\replaced[af]{Omitted}{Missing} proofs from Section \ref{s:alg}}
\label{app:LIS}

\deleted[af]{
Before giving the proof of Theorem~\ref{t:LISs-comp1}, we analyze several properties of the algorithm needed for this proof.
We first focus on the pending-tasks competitiveness.
Suppose that $\LISs$ is not $\OPT + \beta n^2+3n$ competitive in terms of the number of pending tasks, OPT for some $\beta \geq \frac{\lmax}{\lmin}$ and some $s\ge \frac{\lmax}{\lmin}$. Consider an execution witnessing this fact and fix the adversarial pattern $\cA$ associated with it, together with the optimum solution \OPT\ for it. 
%
Let $t^*$ be a time in the execution when $\cT_{t^*}(\LISs,\cA) > \cT_{t^*}(\OPT,\cA) + \beta n^2+3n$. For any time interval $I$, let $\cT_I$ be the total number of tasks injected in the interval $I$. Let $t_*\le t^*$ be the smallest time such that for all $t\in [t_*,t^*)$, $\cT_t(\LISs,\cA) > \cT_t(\OPT,\cA) +  \beta n^2$.
(Note that the selection of minimum time satisfying some properties defined by the computation is possible due to the fact that the computation is split into discrete processor cycles.)
Observe that $\cT_{t_*}(\LISs,\cA) \le \cT_{t_*}(\OPT,\cA) + \beta n^2 +n$, because at time $t_*$ no more than $n$ tasks could be reported to the repository by \OPT, while just before $t_*$ the difference between $\LISs$ and \OPT was at most $\beta n^2$.
}

\begin{lemma}
\label{l:interval}
We have $t_*< t^*-\lmin$, and for every $t\in  [t_*,t_*+\lmin]$ the following holds with respect to the number of pending tasks: $\cT_t(\LISs,\cA) \le \cT_t(\OPT,\cA) + \beta n^2+2n$.
\end{lemma}

\begin{proof}
We already discussed the case $t=t_*$. In the interval $(t_*,t_*+\lmin]$, $\OPT$ can notify the repository about at most $n$ performed tasks, as each of $n$ processors may finish at most one task. Consider any $t\in (t_*,t_*+\lmin]$ and let $I$ be fixed to $(t_*,t]$.
We have $\cT_t(\LISs,\cA) \le \cT_{t_*}(\LISs,\cA)+\cT_I$ and $\cT_t(\OPT,\cA) \ge \cT_{t_*}(\OPT,\cA) + \cT_I - n$.
It follows that 
\begin{eqnarray*}
\cT_t(\LISs,\cA)
&\le& 
\cT_{t_*}(\LISs,\cA)+\cT_I 
\\
&\le&
\left(\cT_{t_*}(\OPT,\cA) +  \beta n^2 + n \right) \\ &~~~+&
\left(\cT_t(\OPT,\cA) - \cT_{t_*}(\OPT,\cA) + n \right)
\\
&\le&
\cT_t(\OPT,\cA) + \beta n^2 +2n
\ .
\end{eqnarray*}
It also follows that any such $t$ must be smaller than $t^*$, by definition of $t^*$.\vspace{.2em}
\end{proof}

\begin{lemma}
\label{l:OPT-not-better}
Consider a time interval $I$ during which the queue of pending tasks in $\LISs$ is always non-empty. Then the total number of tasks reported by $\OPT$ in the period $I$ is not bigger than the total number of tasks reported by $\LISs$ in the same period plus $n$ (counting possible redundancy).
\end{lemma}

\begin{proof}
For each processor in the execution of $\OPT$, under the adversarial pattern $\cA$, in the considered period, exclude the first reported task; this is to eliminate from further analysis tasks that might have been started before time interval $I$. There are at most $n$ such tasks reported by $\OPT$.

It remains to show that the number of remaining tasks reported to the repository by $\OPT$ is not bigger than those reported in the execution of $\LISs$ in the considered period $I$. It follows from the property that $s\ge \frac{\lmax}{\lmin}$. More precisely, it implies that during time period when a processor $p$ performs a task $\tau$ in the execution of $\OPT$,  the same processor reports at least one task to the repository in the execution of $\LISs$. 
This is because performing any task by a processor in the execution of $\OPT$ takes at least time $\lmin$, while performing any task by $\LISs$ takes no more than $\frac{\lmax}{s} \le \lmin$, and also because no active processor in the execution of $\LISs$ is ever idle due to non-emptiness of the pending task queue. Hence we can define a 1-1 function from the considered tasks performed by $\OPT$ (i.e., tasks which are started and reported in time interval $I$) to the family of different tasks reported by $\LISs$ in the period $I$, which completes the proof.
\end{proof}

\begin{lemma}
\label{l:redundancy}
In the interval $(t_*+\lmin,t^*]$ no task is reported twice to the repository by $\LISs$.
\end{lemma}

\begin{proof}
The proof is by contradiction. Suppose that task $\tau$ is reported twice in the considered time interval of the execution of $\LISs$, under adversarial pattern $\cA$. Consider the first two such reports, by processors $p_1$ and $p_2$; w.l.o.g. we may assume that $p_1$ reported $\tau$ at time $t_1$, not later than $p_2$ reported $\tau$ at time $t_2$. Let $\ell_\tau$ denote the cost of task $\tau$.
The considered reports have to occur within time period shorter than the cost of task $\tau$, in particular, shorter than $\lmax/s\le \lmin$; otherwise it would mean that the processor who reported as the second would have started performing this task not earlier than the previous report to the repository, which contradicts the property of the repository that each reported task is immediately removed from the list of pending tasks. It also implies that $p_1\ne p_2$.

From the algorithm description, the list $Pending$ at time $t_1-\ell_\tau/s$ had task $\tau$ at position $p_1 \beta n$, while the list $Pending$ at time $t_2-\ell_\tau/s$ had task $\tau$ at position $p_2 \beta n$.
Note that interval $[t_1-\ell_\tau/s,t_2-\ell_\tau/s]$ is included in $[t_*,t^*]$, and thus, by the definition of $t_*$, at any time of this interval there are at least $\beta n^2$ tasks in the list $Pending$.

There are two cases to consider. First, if $p_1<p_2$, then because new tasks on list $Pending$ are appended at the end of the list, it will never happen that a task with rank $p_1 \beta n$ would increase its rank in time, in particular, not to $p_2 \beta n$.
Second, if $p_1>p_2$, then during time interval $[t_1-\ell_\tau/s,t_2-\ell_\tau/s]$ task $\tau$ has to decrease its rank from $p_1 \beta n$ to $p_2 \beta n$, i.e., by at least $\beta n$ positions. It may happen only if at least $\beta n$ tasks ranked before $\tau$ on the list $Pending$ at time $t_1-\ell_\tau/s$ become reported in the considered time interval.
Since all of them are of cost at least $\lmin$, and the considered time interval has length smaller than $\lmax/s$, each processor may report at most $\frac{\lmax/s}{\lmin/s}\le \beta$ tasks (this is the part of analysis requiring $\beta\ge \frac{\lmax}{\lmin}$). 
Since processor $p_2$ can report at most $\beta-1$ tasks different than $\tau$, the total number of tasks different from $\tau$ reported to the repository is at most $\beta n-1$, and hence it is not possible to reduce the rank of $\tau$ from $p_1 \beta n$ to $p_2 \beta n$ within the considered time interval. This contradicts the assumption that $p_2$ reports $\tau$ to the repository at time $t_2$.
\end{proof}

\deleted[af]{
\begin{lemma}
\label{l:number-tasks}
$\cT_{t^*}(\LISs,\cA) \le \cT_{t^*}(\OPT,\cA) + \beta n^2+3n$.
\end{lemma}
}

\begin{proofof}{Lemma~\ref{l:number-tasks}}
By Lemma~\ref{l:interval} we have that $\cT_{t_*+\lmin}(\LISs,\cA) \le \cT_{t_*+\lmin}(\OPT,\cA) + \beta n^2+2n$.\\
Let $y$ be the total number of tasks reported by $\LISs$ in $(t_*+\lmin,t^*]$.
By Lemma~\ref{l:OPT-not-better} and definitions $t_*$ and $t^*$, $\OPT$ reports no more that $y+n$ tasks in $(t_*+\lmin,t^*]$. Therefore,
\[
\cT_{t^*}(\OPT,\cA)
\ge
\cT_{t_*+\lmin}(\OPT,\cA)-(y+n)
\ .
\]
By Lemma~\ref{l:redundancy}, in the interval $(t_*+\lmin,t^*]$, no redundant work is reported by $\LISs$. Thus,
\[
\cT_{t^*}(\LISs,\cA)
\le
\cT_{t_*+\lmin}(\LISs,\cA)-y
\ .
\]

Consequently,
\begin{eqnarray*}
\cT_{t^*}(\LISs,\cA)
&\le&
\cT_{t_*+\lmin}(\LISs,\cA) - y
\\
&\le&
\left(\cT_{t_*+\lmin}(\OPT,\cA)+\beta n^2+2n\right) - y
\\
&\le&
\cT_{t^*}(\OPT,\cA)+(\beta n^2 +2n) + n
\\
&\le&
\cT_{t^*}(\OPT,\cA)+\beta n^2 + 3n
\end{eqnarray*}
as desired.\vspace{.3em}
\end{proofof}

\deleted[af]{
\noindent We are now ready to prove Theorem~\ref{t:LISs-comp1}.
\begin{proofof}{Theorem~\ref{t:LISs-comp1}}
The competitiveness for the number of pending tasks follows directly from Lemma~\ref{l:number-tasks}: it violates the contradictory assumptions made in the beginning of the analysis.
The result for the pending cost is a direct consequence of the one for pending tasks, as the cost of any pending task in $\LISs$ is at most $\frac{\lmax}{\lmin}$ times bigger than the cost of any pending task in $\OPT$.
\end{proofof}
}

\section{\replaced[af]{Omitted}{Missing} proofs from Section \ref{sec:gnBurst}}
\label{a:gnburst}

We begin the analysis of $\gnBurst$ with necessary definitions.

\begin{definition}
We define the {\bf\em absolute task execution} of a task $\tau$ to be the interval $[t,t^\prime]$ in which a processor $p$ schedules $\tau$ at time $t$ and reports its completion to the repository at $t^\prime$, without stopping its execution within the interval $[t,t^\prime)$.
\end{definition}

\begin{definition}
We say that a scheduling algorithm is of type {\bf\em \GroupLIS$\mathbf{(\beta)}$}, $\beta\in\mathbb{N}$, if all the following hold:
\begin{itemize}[leftmargin=5mm]
\item It classifies the pending tasks into classes where each class contains tasks of the same cost.
\item It sorts the tasks in each class in increasing order with respect to their arrival time.
\item If a class contains at least $\beta\cdot n^2$ pending tasks and a processor $p$ schedules a task from that class, then it schedules
the $(p\cdot \beta n)$th task in the class.
\end{itemize} 
\end{definition}

Observe that algorithm  $\gnBurst$ is of type \GroupLIS$(1)$.
The next lemmas state useful properties of algorithms of type \GroupLIS.\vspace{.3em}

\begin{lemma}
\label{l:redundancy2}
For an algorithm $A$ of type \mbox{\GroupLIS$(\beta)$} and a time interval $I$ in which a list $L$ of tasks of cost $c$ has at least $\beta\cdot n^2$ pending tasks, any two absolute task executions fully contained in $I$, of tasks $\tau_1,\tau_2 \in L$, by processors $p_1$ and $p_2$ respectively, must have $\tau_1 \neq \tau_2$.\vspace{.3em}
\end{lemma}

\begin{proof}
Suppose by contradiction, that two processors $p_1$ and $p_2$ schedule the same $c$-task, say $\tau \in L$, to be executed during the interval $I$. Let's assume times $t_1$ and $t_2$, where $t_1,t_2 \in I$ and $t_1 \leq t_2$, to be the times when each of the processors correspondingly, scheduled the task. Since any $c$-task takes time $\frac{c}{s}$ to be completed, then $p_2$ must schedule the task before time $t_1+\frac{c}{s}$, or else it would contradict the property of the Dispatcher stating that each reported task is immediately removed from the set of pending tasks.\\
Since algorithm $A$ is of type \GroupLIS$(\beta)$, we have that at time $t_1$, when $p_1$ schedules $\tau$, the task's position on the list $L$ is $p_1\cdot \beta n$. In order for processor $p_2$ to schedule $\tau$ at time $t_2$, it must be at position $p_2\cdot \beta n$. There are two cases we have to consider:\\
(1) If $p_1 < p_2$, then during the interval $[t_1,t_2]$, task $\tau$ must increase its position in the list $L$ from $p_1\cdot \beta n$ to $p_2\cdot \beta n$, i.e., by at least $\beta n$ positions. This can happen only in the case where new tasks are injected and are placed before $\tau$. This, however, is not possible, since new $c$-tasks are appended at the end of the list. (Recall that in algorithms of type \GroupLIS, the tasks in $L$ are 
sorted in an increasing order with respect to arrival times.)\\
(2) If $p_1 > p_2$, then during the interval $[t_1,t_2]$, task $\tau$ must decrease its position in the list by at least $\beta n$ places. This may happen only in the case where at least $\beta n$ tasks ordered before $\tau$ in $L$ at time $t_1$, are completed and reported by time $t_2$. Since all tasks in list $L$ are of the same cost $c$, and the considered interval has length $\frac{c}{s}$, each processor may complete at most one task during that time. Hence, at most $n-1$ $c$-tasks may be completed, which are not enough to change $\tau$'s position from $p_1\cdot \beta n$ to $p_2\cdot \beta n$ (even when $\beta=1$) by time $t_2$.\\
The two cases above contradict the initial assumption and hence the claim of the lemma follows.
\end{proof}

\begin{lemma}
\label{l:nfirst}
Let $S$ be a set of tasks reported as completed by an algorithm $A$ of type \mbox{\GroupLIS$(\beta)$} in a time interval $I$. Then at least $|S| - n$ such tasks have their absolute task execution fully contained in $I$.\vspace{.3em}
\end{lemma}

\begin{proof}
A task $\tau$ which is reported in $I$ by processor $p$ and its absolute task execution $\alpha \not\subseteq I$, has $\alpha = [t,t^\prime]$ where $t \not\in I$ and $t^\prime \in I$. Since $p$ does not stop executing $\tau$ in $[t,t^\prime)$, only one such task may occur for $p$. Then, overall there can not be more than $n$ such reports and the lemma follows.
\end{proof}

Consider the following two interval types, used in the remainder of the section. $\cT^{\max}_t(A,\cA)$ and $\cT^{\min}_t(A,\cA)$ denote the number of pending tasks at time $t$ with algorithm $A$ of costs $\lmax$ and $\lmin$, respectively,
under adversarial pattern $\cA$.
Consider two types of intervals:
\begin{itemize}
\item [$I^+$:] any interval such that $\cT^{\max}_t(\gnBurst,\cA) \geq n^2$, $\forall t \in I^+$
\item [$I^-$:] any interval such that $\cT^{\min}_t(\gnBurst,\cA) \geq n^2$, $\forall t \in I^-$
\end{itemize}

Then, the next two lemmas follow from Lemma~\ref{l:redundancy2} and that algorithm $\gnBurst$ is of type \GroupLIS$(1)$.\vspace{.3em}

\begin{lemma}
\label{l:max-no-red}
All absolute task executions of $\lmax$-tasks in Algorithm $\gnBurst$ within any interval $I^+$ appear exactly once.\vspace{.2em}
\end{lemma}

\begin{lemma}
\label{l:min-no-red}
All absolute task executions of $\lmin$-tasks in Algorithm $\gnBurst$ within any interval $I^-$ appear exactly once.
\end{lemma}

The above leads to the following upper bound on the difference in the number of pending $\lmax$-tasks.\vspace{.3em}

\begin{lemma}
\label{l-ytasks}
The number of pending $\lmax$-tasks in any execution of $\gnBurst$, under any adversarial pattern $\cA$, run with speed-up $s \geq \frac{\gamma \lmin+\lmax}{\lmax}$, is never larger than the number of pending $\lmax$-tasks in the execution of \OPT plus $n^2+2n$.\vspace{.3em}
\end{lemma}

\begin{proof}
Fix an adversarial pattern $\cA$ and consider, for contradiction, interval $I^+ = (t_*,t^*]$ as it was defined above, $t^*$ being the first time when 
$\cT^{\max}_{t^*}(\gnBurst,\cA) > \cT^{\max}_{t^*}(\OPT,\cA) + n^2 + 2n$, and $t_*$ being the largest time before $t^*$ such that $\cT^{\max}_{t_*}(\gnBurst,\cA) < n^2$.\vspace{.2em}

\emph{Claim:} The number of absolute task executions of $\lmax$-tasks $\alpha \subset I^+$, by \OPT, is no bigger than the number of $\lmax$-task reports by $\gnBurst$ in interval $I^+$. \vspace{.2em}

Since $s \geq \frac{\gamma \lmin + \lmax}{\lmax}$, while processor $p$ in \OPT is running a $\lmax$-task, the same processor in $\gnBurst$ has time to execute $\gamma \lmin + \lmax$ tasks. But, by definition, within the interval $I^+$ there are at least $n^2$ $\lmax$-task pending at all times, which implies the execution of Case 3 or Case 4 of the $\gnBurst$ algorithm. This means that no processor may run $\gamma + 1$ consecutive $\lmin$-tasks, as a $\lmax$-task is guaranteed to be executed by one of the cases. So, the number of absolute task executions of $\lmax$-tasks by \OPT in the interval $I^+$ is no bigger than the number of $\lmax$-task reports by $\gnBurst$ in the same interval. This completes the proof of the claim.\vspace{.2em}

Now let $\kappa$ be the number of $\lmax$-tasks reported by \OPT. From Lemma \ref{l:nfirst}, at least $\kappa - n$ such tasks have absolute task executions in interval $I^+$. From the above claim, for every absolute task execution of $\lmax$-tasks in the interval $I^+$ by \OPT, there is at least a completion of a $\lmax$-task by $\gnBurst$ which gives a 1-1 correspondence, so $\gnBurst$ has at least $\kappa - n$ reported $\lmax$-tasks in $I^+$.
Also, from Lemma \ref{l:nfirst}, we may conclude that there are at least $\kappa - 2n$ absolute task executions of $\lmax$-tasks in the interval. 
Then from Lemma \ref{l:redundancy2}, $\gnBurst$ reports at least $\kappa - 2n$ different tasks, while \OPT reports at most $\kappa$.

Now let $S_{I^+}$ be the set of $\lmax$-tasks injected during the interval $I^+$, under adversarial pattern $\cA$. Then $\cT^{\max}|_{t^*}(\gnBurst,\cA) < n^2 + |S_{I^+}| - (\kappa - 2n)$, and since $\cT^{\max}_{t^*}(\OPT,\cA) \geq |S_{I^+}| - \kappa$ we have a contradiction, which completes the proof.
\end{proof}


\begin{proofof}{Theorem~\ref{t:gamma-b-com}}
Consider any adversarial pattern $\cA$ and for contradiction, the interval $I^- = (t_*,t^*]$ as defined above, where $t^*$ is the first time when 
$\cT_{t^*}(\gnBurst,\cA) > \cT_{t^*}(\OPT,\cA) + 2n^2 + (3+\left\lceil \frac{\lmax}{s  \cdot \lmin} \right\rceil)n$
and $t_*$ being the largest time before $t^*$ such that $\cT^{\max}|_{t_*}(\gnBurst,\cA) < n^2$. Notice that $t_*$ is well defined for Lemma \ref{l-ytasks},
i.e., such time $t_*$ exists and it is smaller than $t^*$.

We consider each processor individually and break the interval $I^-$ into subintervals $[t,t']$ such that times $t$ and $t'$ are instances in which the counter $c$ is reset to 0; this can be either due to a simple reset in the algorithm or due to a crash and restart of a processor. 
More concretely, the boundaries of such subintervals are as follows. An interval can start either when a reset of the counter occurs or when the processor (re)starts.
On its side, an interval can finish due to either a reset of the counter or a processor crash.
Hence, these subintervals can be grouped into two types, depending on how they end: Type (a) which includes the ones that end by a crash and Type (b) which includes the ones that end by a reset from the algorithm.
Note that in all cases $\gnBurst$ starts the subinterval scheduling a new task to the processor at time $t$, and that the processor is never idle in the interval. Hence, all tasks reported by $\gnBurst$ as completed have their absolute task execution completely into the subinterval.
Our goal is to show that the number of absolute task executions in each such subinterval with $\gnBurst$ is no less than the number of reported tasks by \OPT.

First, consider a subinterval $[t,t']$ of Type (b), that is, such that the counter $c$ is set to 0 by the algorithm (in a line $c=0$) at time $t'$. 
This may happen in algorithm $\gnBurst$ in Cases 1, 3 or 4. However, observe that the counter cannot be reset in Cases 1 and 3
at time $t' \in I^-$ since,
by definition, there are at least $n^2$ $\lmin$-tasks pending during the whole interval $I^-$. Case 4 implies that there are also at least $n^2$ $\lmax$-tasks pending in $\gnBurst$. This means that in the interval $[t,t']$ there have been 
$\kappa$ $\lmin$ and one $\lmax$ absolute task executions, $\kappa \geq \gamma$. Then, the subinterval $[t,t']$ has length $\frac{\lmax + \kappa\lmin}{s}$, and $\OPT$ can report at most $\kappa+1$ task completions during the subinterval. This latter property follows from 
$
\frac{\lmax + \kappa\lmin}{s} = \frac{\lmax + \gamma\lmin}{s} + \frac{(\kappa - \gamma)\lmin}{s}
\leq
(\gamma + 1)\lmin + (\kappa - \gamma)\lmin
\leq
(\kappa + 1)\lmin,
$
where the first inequality follows from the definition of $\gamma$ (see Section~\ref{s:non-comp}) and the fact that $s>1$.
Now consider a subinterval $[t,t']$ of Type (a) which means that at time $t'$ there was a crash. This means that no $\lmax$-task was completed in the subinterval, but we may assume the complete execution of $\kappa$ $\lmin$-tasks in $\gnBurst$. 
We show now that $\OPT$ cannot report more than
$\kappa$ task completions.
In the case where $\kappa \geq \gamma$, then the length of the subinterval $[t,t']$ satisfies
\begin{eqnarray*}
t' - t
<
\frac{\kappa\lmin + \lmax}{s}
&\leq&
(\kappa + 1)\lmin.
\
\end{eqnarray*}
In the case where $\kappa < \gamma$ then the length of the subinterval $[t,t']$ satisfies
\begin{eqnarray*}
t' - t
<
\frac{(\kappa + 1)\lmin}{s}
&\leq&
(\kappa + 1)\lmin.
\
\end{eqnarray*}
Then in none of the two cases $\OPT$ can report more than $\kappa$ tasks in subinterval $[t,t']$.

After splitting $I^-$ into the above subintervals, the whole interval is of the form $(t_*,t_1][t_1,t_2]\dots[t_m,t^*]$. All the intervals $[t_i,t_{i+1}]$ where $t=1,2,\dots,m$, are included in the subinterval types already analysed. There are therefore two remaining subintervals to consider now.
The analysis of subinterval $[t_m,t^*]$ is verbatim to that of an interval of Type (a). Hence, the number of absolute task executions in that subinterval with $\gnBurst$ is no less than the number of reported tasks by \OPT.

Let us now consider the subinterval $(t_*,t_1]$. Assume with $\gnBurst$ there are 
$\kappa$ absolute task executions fully contained in the subinterval. Also observe that at most one $\lmax$-task
can be reported in the subinterval (since then the counter is reset and the subinterval ends). Then, the
length of the subinterval is bounded as
$$
t_1 - t_* < \frac{(\kappa + 1)\lmin + \lmax}{s}
$$
(assuming the worst case that a $\lmin$-task was just started at $t_*$ and that the processor crashed at $t_1$ when a $\lmax$-task was about to finish).
The number of tasks that $\OPT$ can report in the subinterval is hence bounded by 
$$
\left\lceil \frac{(\kappa + 1)\lmin + \lmax}{s \lmin} \right\rceil < \kappa + 1 + \left\lceil \frac{\lmax}{s \cdot \lmin} \right\rceil.
$$

This means that for every processor, the number of reported tasks by \OPT might be at most the number of absolute task executions by $\gnBurst$ fully contained in $I^-$ plus $1 + \left\lceil \frac{\lmax}{s  \cdot \lmin} \right\rceil$.
From this and Lemma~\ref{l:min-no-red}, it follows that in interval $I^-$ the difference in the number of pending tasks between
$\gnBurst$ and $\OPT$ has grown by at most $(1+\left\lceil \frac{\lmax}{s \cdot \lmin} \right\rceil)n$.
Observe that at time $t_*$ the difference between the number of pending tasks satisfied
$$
\cT_{t_*}(\gnBurst,\cA) - \cT_{t_*}(\OPT,\cA) < 2n^2 + 2n,
$$
This follows from Lemma~\ref{l-ytasks}, which bounds the difference in the number of $\lmax$-tasks to $n^2 + 2n$, and the assumption that $\cT^{\max}|_{t_*}(\gnBurst,\cA) < n^2$. Then, it follows that
$
\cT_{t^*}(\gnBurst,\cA) - \cT_{t_*}(\OPT,\cA)  <  2n^2 + 2n + (1+\left\lceil \frac{\lmax}{s  \cdot \lmin} \right\rceil)n
 =  n^2 + (3+\left\lceil \frac{\lmax}{s  \cdot \lmin} \right\rceil)n,
$
which is a contradiction. Hence,
$\cT_t(\gnBurst,\cA) \leq \cT_t(\OPT,\cA) + 2n^2 + (3+\left\lceil \frac{\lmax}{s  \cdot \lmin} \right\rceil)n$,  for any time  $t$ and adversarial pattern $\cA$, as claimed.
\end{proofof}

\section{\replaced[af]{Omitted}{Missing} proofs from Section \ref{sec:LAF}}
\label{app:LAF}

\begin{proofof}{Theorem~\ref{thm:LAF}}
Note that algorithm $\LAF$ is in the class of $\GroupLIS(\beta)$ algorithms, for $\beta\geq \frac{\lmax}{\lmin}$.
Therefore Lemma~\ref{l:redundancy2} applies, and together with the algorithm specification it guarantees no redundancy in 
absolute task executions in case of one of the lists is kept of size at least $\beta n^2$.

Consider any adversarial pattern $\cA$.
We show now that 
$\cC^*_t(\LAF,\cA)|_{\ge x}\le \cC^*_t(\OPT,\cA)|_{\ge x}+2\lmax k \beta n^2 +2n\lmax+3n\lmax/s$
for every cost $x$ at any time $t$ and for speedup $s$,
where $\cC^*_t(\Alg,\cA)|_{\ge x}$ denotes the sum of costs of pending tasks of cost at least $x$, 
and such that the number of pending tasks of such cost is at least $\beta n^2$ in $\LAF$ at time $t$
of the execution of algorithm $\Alg$, 
under adversarial pattern $\cA$; $k$ is the number of the possible different task costs that
is injected under adversarial pattern $\cA$.
Note that this implies the statement of the theorem, since if we take $x$ equal to the smallest
possible cost and add an upper bound $\lmax k \beta n^2$ on the cost of tasks on pending lists of $\LAF$ of size smaller than 
$\beta n^2$, we obtain the upper bound on the amount of pending cost of $\LAF$,
for any adversarial pattern $\cA$.

Assume, to the contrary, that the sought property does not hold, and let $t^*$ will be the first time $t$ when 
$\cC^*_t(\LAF,\cA)|_{\ge x} > \cC^*_t(\OPT,\cA)|_{\ge x}+2\lmax k \beta n^2 + 2n\lmax+3n\lmax/s$
for some cost $x$, under the adversarial pattern $\cA$ (in the remainder of the proof
we work under assumption of the fixed adversarial pattern $\cA$). 
Denote by $t_*$ the largest time before $t^*$ such that for every $t\in (t_*,t^*]$,
$\cC^*_t(\LAF,\cA)|_{\ge x} \ge \cC^*_t(\OPT,\cA)|_{\ge x}+\lmax k\beta n^2$.
Observe that $t_*$ is well-defined, and moreover, $t_*\le t^*-(\lmax+3\lmax/s)$:
it follows from the definition of $t^*$ and from the fact that within a time interval $(t,t^*]$ of length smaller than 
$\lmax+3\lmax/s$, $\OPT$ can report tasks of total cost at most 
$2n\lmax+3n\lmax/s$,
plus additional cost of at most $\lmax k\beta n^2$ that can be caused by other
lists growing beyond the threshold $\beta n^2$, and thus starting to contribute to the cost $\cC^*$.

Consider interval $(t_*,t^*]$.
By the specification of $t_*$, at any time of the interval there is at least one list of pending tasks of cost at least $x$
that has length at least $\beta n^2$.
Consider a life period of a process $p$ that starts in the considered time interval; let us restrict our consideration
of this life period only by time $t^*$, and $\ell$ be the length of this period.
Let $z>0$ be the total cost of tasks, when counted only those of cost at least $x$, reported by processor $p$ in the execution of $\OPT$
in the considered life period.
We argue that in the same time interval, the total cost of tasks, when counted only those of cost at least $x$,
reported by $p$ in the execution of $\LAF$ is at least $z$.
Observe that once process $p$ in $\LAF$ schedules a task of cost at least $x$ for the first time in the considered period,
it continues scheduling task of cost at least $x$ until the end of the considered period.
Therefore, with respect to the corresponding execution of $\OPT$, processor $p$ could only waste its time
(from perspective of performing a task of cost smaller than $x$ or performing a task not reported in the considered period)
in the first less than $(2x)/s$ time of the period or the last less than $(\ell/2)/s$ time of the period.
Therefore, in the remaining period of length bigger than $\ell-(\ell/2+2x)/s$, processor $p$ is able to complete and report
tasks, each of cost at least $x$, of total cost larger than
\[
s\ell - (\ell/2+2x)
\ge 
\ell (s-1/2-2)
\ge 
\ell
\ge
z
\ ;
\]
here in the first inequality we used the fact that 
$\ell\ge x$, 
which follows from the definition of $z>0$,
and in the second inequality we used the property 
$s-1/2-2\ge 1$ for $s\ge 7/2$.
Applying Lemma~\ref{l:redundancy}, justifying no redundancy in absolute tasks executions of $\LAF$ in the
considered time interval,
we conclude life periods as considered do not contribute to the growth of the difference between 
$\cC^*(\LAF,\cA)|_{\ge x}$ and $\cC^*(\OPT,\cA)|_{\ge x}$.

Therefore, only life periods that start before $t_*$ can contribute to the difference in costs.
However, if their intersections with the time interval $(t_*,t^*]$ is of length $\ell$ at least $(2x+\lmax)/s$,
that is, enough for a processor running $\LAF$ to report at least one task of length at least $x$,
the same argument as in the previous paragraph yields that the total cost of tasks of cost at least $x$ reported
by a processor in the execution of $\LAF$ is at least as large as in the execution of $\OPT$,
minus the cost of the very first task reported by each processor in $\LAF$ (which may not be an absolute task execution
and thus there may be redundancy on them) --- i.e., minus at most $n\lmax$ in total.
In the remaining case, i.e., when the intersection of the life period with $(t_*,t^*]$ is smaller than $(2x+\lmax)/s$,
the processor may not report any task of length $x$ when running $\LAF$, but when executing $\OPT$
the total cost of all reported tasks is smaller than $(2x+\lmax)/s\le 3\lmax/s$.
Therefore, the difference in costs on tasks of cost at least $x$ between $\OPT$ and $\LAF$ could grow by
at most $n\lmax+3n\lmax/s$ in the life periods considered in this paragraph. 
Hence,
$
\cC^*_{t^*}(\LAF,\cA)|_{\ge x} - \cC^*_{t^*}(\OPT,\cA)|_{\ge x}
\le 
\cC^*_{t_*}(\LAF,\cA)|_{\ge x} - \cC^*_{t_*}(\OPT,\cA)|_{\ge x} +n\lmax+3n\lmax/s
\le 
\lmax k\beta n^2 + n\lmax+3n\lmax/s,
$
which violates the initial contradictory assumption.
\end{proofof}

\section{Conditions on Competitiveness and Non-competitiveness}
\label{app:conditions}

\paragraph{Upper bound on the speedup for non-competitiveness}
As proven in Theorem~\ref{t:non-competitive}, the condition $s < \min\left\{ \frac{\lmax}{\lmin}, \frac{\gamma \lmin+\lmax}{\lmax}\right\}$ 
is sufficient for non competitiveness. 
Let us define ratio $\rho=\lmax/\lmin \geq 1$. We will derive properties in $\rho$ that guarantee the above condition.
From the first part (condition (a) in Theorem~\ref{t:non-competitive}), it must hold that $s < \frac{\lmax}{\lmin} = \rho$.
From the second part (condition (b) in Theorem~\ref{t:non-competitive}), we must have
\begin{eqnarray}
\nonumber
s &<& \frac{\gamma \lmin+\lmax}{\lmax} \\
&= &
\frac{\lceil \frac{\lmax-s \lmin}{\lmin(s-1)} \rceil \lmin+\lmax}{\lmax} \nonumber
\\
&=& 
\frac{\lceil \frac{\lmax-\lmin}{\lmin(s-1)} \rceil \lmin+\lmax-\lmin}{\lmax} \nonumber \\
&=&
\frac{\lceil \frac{\rho - 1}{s-1} \rceil + \rho -1}{\rho}
\ ,
\label{condba}
\end{eqnarray}
where the second equality follows from
$
\lceil \frac{\lmax-s\lmin}{\lmin(s-1)} \rceil 
= 
\lceil \frac{\lmax-\lmin}{\lmin(s-1)} \rceil -1
$.
Let $s_{\ref{condba}}$ be the smallest speedup that satisfies Eq.~\ref{condba}, 
then a lower bound on $s_{\ref{condba}}$ can be found by removing the ceiling, as
$$
s_{\ref{condba}} \geq \frac{\frac{\rho - 1}{s_b -1} + \rho -1}{\rho}  
\implies 
s_{\ref{condba}} \geq 2 - 1/\rho.
$$
It can be shown that $\rho \geq 2 - 1/\rho$ for $\rho \geq 1$. Then, a sufficient condition for non competitiveness is
$$
s < 2 - 1/\rho = 2 - \lmin/\lmax.
$$

\paragraph{Smallest speedup for competitiveness}

As we show in this work, in order to have competitiveness, $s \geq \min\left\{ \frac{\lmax}{\lmin}, \frac{\gamma \lmin+\lmax}{\lmax}\right\}$ is sufficient. This means that
(a) $s \geq \frac{\lmax}{\lmin}$, or (b) $s \geq \frac{\gamma\lmin+\lmax}{\lmax}$
must hold, where $\gamma = \max\{\lceil \frac{\lmax-s \lmin}{(s-1)\lmin} \rceil, 0\}$.
To satisfy condition (a), the
speedup $s$ must satisfy $s \geq \frac{\lmax}{\lmin} = \rho$.
Hence, the smallest value of $s$ that guarantees that (a) holds is $s_{(a)}=\rho$.

In order to satisfy condition (b), when condition (a) is not satisfied
(observe that when (a) holds, $\gamma=0$),
we have
\begin{eqnarray}
s
&\geq&
\frac{\lceil \frac{\rho - 1}{s-1} \rceil + \rho -1}{\rho}
\ .
\label{condbb}
\end{eqnarray}
Let $s_{(b)}$ be the smallest speedup that satisfies Eq.~\ref{condbb};
then an upper bound can be obtained by adding one unit to the expression in the ceiling
$$
s_{(b)} < \frac{\frac{\rho - 1}{s_{(b)}-1} +1 + \rho -1}{\rho}
\implies
s_{(b)} < 1 + \sqrt{1 - 1/\rho}
\ .
$$
Let us denote $s_{(b)}^+ = 1 + \sqrt{1 - 1/\rho}$. Then,
in order to guarantee competitiveness, it is enough to choose any 
$s \geq \min\{s_{(a)} , s_{(b)}\}$.
Since there is no simple form of the
expression for $s_{(b)}$, we can use $s_{(b)}^+$ instead, to be safe.

\begin{theorem}
Let $\rho=\lmax/\lmin \geq 1$.
In order to have competitiveness, it is sufficient to set $s=s_{(a)}=\rho$ if $\rho \in [1,\varphi]$, and
$s=s_{(b)}^+ = 1 + \sqrt{1 - 1/\rho}$ if $\rho > \varphi$, where
$\varphi=\frac{1+\sqrt{5}}{2}$ is the golden ratio.
\end{theorem}

\begin{proof}
As mentioned before, a sufficient condition for competitiveness is 
$s \geq \min\{s_{(a)} , s_{(b)}^+\}$.
Using calculus is it easy to verify that
$s{(a)}=\rho \leq s_{(b)}^+$ if $\rho \leq \varphi$. 
\end{proof}

\end{document}